%% file: main.tex
\documentclass[11pt]{article}
\input{header.tex}

\newcommand{\ElimNeg}{\operatorname{ElimNeg}}
\newcommand{\SPaverage}{\ElimNeg}

\renewcommand{\sin}{s_{in}}
\newcommand{\gin}{\Gin}

\newcommand{\win}{w_{in}}

\newcommand{\tspmain}{\mathcal{T}_{spmain}}

\newcommand{\dist}{\operatorname{dist}}

\newcommand{\Otil}{\tilde{O}}

\newcommand{\cA}{\mathcal{A}}

\newcommand{\Dijkstra}{\operatorname{Dijkstra}}

\newcommand{\SPmain}{\operatorname{SPmain}} 
\newcommand{\ScaleDown}{\operatorname{ScaleDown}}

\newcommand{\SCCDecomposition}{\operatorname{LowDiamDecomposition}}
\newcommand{\LowDiamDecomposition}{\SCCDecomposition}

\newcommand{\FixAlmostDag}{\operatorname{FixDAGEdges}}
\newcommand{\LDD}{\operatorname{LDD}}

\newcommand{\eneg}{E^{neg}}
\newcommand{\esep}{E^{rem}}
\newcommand{\erecurse}{E^{recurse}}
\newcommand{\eboundary}{E^{boundary}}
\newcommand{\disthat}{\widehat{\dist}}
\newcommand{\wstar}{w^*}
\newcommand{\Gstar}{G^*}

\newcommand{\Gin}{G_{in}}
\newcommand{\Ein}{E_{in}}
\newcommand{\Ghat}{\hat{G}}
\newcommand{\Gbar}{\bar{G}}

\newcommand{\what}{\hat{w}}
\newcommand{\wbar}{\bar{w}}

\newcommand{\GB}{G^{B}}
\newcommand{\wB}{w^{B}}
\newcommand{\HB}{H^B}

\DeclareMathOperator{\Geo}{Geo}

\newcommand{\SPLasVegas}{\operatorname{SPLasVegas}}
\newcommand{\SPMonteCarlo}{\operatorname{SPMonteCarlo}}
\newcommand{\FindThresh}{\operatorname{FindThresh}}

\newcommand{\ball}{\operatorname{Ball}}

\newcommand{\ingball}{\operatorname{Ball_G^{in}}}
\newcommand{\outgball}{\operatorname{Ball_G^{out}}}
\newcommand{\stargball}{\operatorname{Ball_G^{*}}}

\newcommand{\ingzball}{\operatorname{Ball_{G_0}^{in}}}
\newcommand{\outgzball}{\operatorname{Ball_{G_0}^{out}}}
\newcommand{\stargzball}{\operatorname{Ball_{G_0}^{*}}}

\newcommand{\boundary}{\partial}

\title{Negative-Weight Single-Source Shortest Paths in Near-linear Time}

\author{Aaron Bernstein\thanks{Rutgers University} \and Danupon Nanongkai\thanks{Max Planck Institute for Informatics and KTH. Work done while at BARC, University of Copenhagen.} \and Christian Wulff-Nilsen\thanks{Work done while at BARC, University of Copenhagen.}}

\date{}
\begin{document}

\maketitle

\pagenumbering{gobble}

\begin{abstract}
We present a randomized algorithm that computes single-source shortest paths (SSSP) in $O(m\log^8(n)\log W)$ time when edge weights are integral and can be negative.\footnote{Throughout, $n$ and $m$ denote the number of vertices and edges, respectively, and $W\geq 2$ is such that every edge weight is at least $-W$. $\tilde O$ hides polylogarithmic factors.} This essentially resolves the classic negative-weight SSSP problem.
The previous bounds are $\tilde O((m+n^{1.5})\log W)$ [BLNPSSSW FOCS'20] and $m^{4/3+o(1)}\log W$ [AMV FOCS'20]. Near-linear time algorithms were known previously only for the special case of planar directed graphs [Fakcharoenphol and Rao FOCS'01].

In contrast to all recent developments that rely on sophisticated continuous optimization methods and dynamic algorithms, our algorithm is simple: it requires only a simple graph decomposition and elementary combinatorial tools. In fact, ours is the first combinatorial algorithm for negative-weight SSSP to break through the classic $\Otil(m\sqrt{n}\log W)$ bound from over three decades ago [Gabow and Tarjan SICOMP'89].

\end{abstract}

\clearpage
\tableofcontents
\clearpage

\pagenumbering{arabic}

\input{Intro}

\input{prelim}

\input{sec_algorithms.tex}

\input{sec_scaledown.tex}

\input{sec_spmain.tex}

\input{negative_cycle}
\input{LDD_new_AB}

\subsection*{Acknowledgement}
We thank Sepehr Assadi, Joakim Blikstad, Thatchaphol Saranurak, and Satish Rao for their comments on the early draft of the paper and the discussions. We thank Mohsen Ghaffari, Merav Parter, Satish Rao, and Goran Zuzic for pointing out some literature on low-diameter decomposition. A discussion with Mohsen Ghaffari led to some simplifications of our low-diameter decomposition algorithm.
We also thank Vikrant Ashvinkumar, Nairen Cao, Christoph Grunau, and Yonggang Jiang for pointing out an earlier error in the paper: in Phase 3, we originally ran $\ElimNeg$ on graph $(G^B_{\phi_2})_s$ and claimed that this graph is equivalent to $G^B_s$; but this is not true because the price function is applied at the wrong time in the order of operations. We thus instead run $\ElimNeg$ on $(G^B_{s})_{\phi_2}$, which is indeed equivalent to $G^B_s$. 

\paragraph{Funding:}
     This project has received funding from the European Research
        Council (ERC) under the European Union's Horizon 2020 research
        and innovation programme under grant agreement No
        715672. Nanongkai was also partially supported by the Swedish
        Research Council (Reg. No. 2019-05622).
        
        Wulff-Nilsen was supported by the Starting Grant 7027-00050B from the
Independent Research Fund Denmark under the Sapere Aude research career programme.

Bernstein is supported by NSF CAREER grant 1942010, a Sloan Fellowship, and the Google Research Scholar Program.

\bibliographystyle{alpha}
 \bibliography{refs}

\newpage
\begin{appendices}
\input{new-appendix-spaverage.tex}

\input{appendix-FixDAG}
\input{las-vegas.tex}

\input{appendix-eboundary}

\end{appendices}

\end{document}

%% file: header.tex
\usepackage{amsmath,amsthm,nicefrac,amssymb}
\usepackage{amsfonts, amstext}

\usepackage[utf8]{inputenc}
\usepackage{enumitem}
\usepackage{geometry}
\usepackage[procnumbered,ruled,vlined,linesnumbered,algo2e]{algorithm2e}

\usepackage{algpseudocode}

\usepackage{amsmath,amsfonts,amsthm,amssymb,dsfont}

\usepackage{color}
\usepackage{graphicx}
\usepackage[title]{appendix}
\usepackage{thmtools,thm-restate}

\usepackage{xcolor}
\usepackage{nameref}
\definecolor{ForestGreen}{rgb}{0.1333,0.5451,0.1333}
\definecolor{DarkRed}{rgb}{0.65,0,0}
\definecolor{Red}{rgb}{1,0,0}
\usepackage[linktocpage=true,
pagebackref=true,colorlinks,
linkcolor=DarkRed,citecolor=ForestGreen,
bookmarks,bookmarksopen,bookmarksnumbered]{hyperref}

\usepackage{cleveref}

\newcommand{\eat}[1]{}

\ifdefined\ShowComment

\usepackage{showlabels}

\def\danupon#1{{\sf\textcolor{pink}{DN: #1}}}
\def\aaron#1{{\sf\textcolor{orange}{AB: #1}}}
\def\christian#1{{\sf\textcolor{blue}{CW: #1}}}
\else

\def\danupon#1{}
\def\aaron#1{}
\def\christian#1{}

\fi

\declaretheorem[numberwithin=section]{theorem}
\declaretheorem[numberlike=theorem]{lemma}

\declaretheorem[numberlike=theorem]{corollary}
\declaretheorem[numberlike=theorem]{claim}
\declaretheorem[numberlike=theorem]{observation}
\declaretheorem[numberlike=theorem]{invariant}
\declaretheorem[numberlike=theorem]{remark}

\crefname{algorithm}{Algorithm}{Algorithms}
\Crefname{algorithm}{Algorithm}{Algorithms}

\theoremstyle{definition}
\declaretheorem[numberlike=theorem]{definition}
\theoremstyle{definition}
\declaretheorem[numberlike=theorem]{assumption}

\usepackage{xparse}             %
\usepackage{xspace}              %
\usepackage[mathscr]{euscript} %

\usepackage{mleftright}         %
\usepackage{mathbbol}           %
\usepackage{fifo-stack}
\usepackage{thmtools}
\usepackage{thm-restate}
\usepackage{letltxmacro}        %
\usepackage{xpatch}              %

\usepackage[T1]{fontenc}%
\usepackage[utf8]{inputenc}%
\usepackage{xparse}
\usepackage{framed}
\usepackage{comment}

\usepackage[export]{adjustbox}
\usepackage{xparse} %
\usepackage{environ} %
\usepackage{xstring}
\usepackage{tabularx}
\usepackage{enumitem}
\usepackage[htt]{hyphenat}%
\usepackage{varwidth}     %
\usepackage{needspace}          %

\NewDocumentCommand{\cutsize}{O{\delta} g g e{_}}{%
  #1
  \IfNoValueF{#3}{_{#3}}%
  \IfNoValueF{#4}{_{#4}}%
  \IfNoValueF{#2}{\parof{#2}}
}%

\newcommand{\ignore}[1]{}

\renewcommand{\paragraph}[1]{\medskip\noindent{\bf #1}\xspace}

%% file: Intro.tex
\section{Introduction} \label{sec:intro}

We consider the single-source shortest paths (SSSP) problem with (possibly negative) integer weights. Given an $m$-edge $n$-vertex directed weighted graph $G=(V,E, w)$ with integral edge weight $w(e)$ for every edge $e\in E$ and a source vertex $s\in V$, we want to compute the distance from $s$ to $v$, denoted by $\dist_G(s,v)$, for every vertex in $v$. 

Two textbook algorithms for SSSP are Bellman-Ford and Dijkstra's algorithm. Dijkstra's algorithm is near-linear time ($O(m+n\log n)$ time), but restricted to {\em nonnegative} edge weights.\footnote{In the word RAM model, Thorup improved the runtime to $O(m+n\log\log(C))$ when $C$ is the maximal edge weight \cite{Thorup04} and to linear time for {\em undirected} graphs \cite{Thorup99}.} 
With negative weights, we can use the Bellman-Ford algorithm, which only requires that there is no {\em negative-weight cycle} reachable from $s$ in $G$; in particular, the algorithm either returns $\dist_G(s,v)\neq -\infty$ for every vertex $v$ or reports that there is a cycle reachable from $s$ whose total weight is negative. Unfortunately, the runtime of Bellman-Ford is $O(mn)$. 

Designing faster algorithms for SSSP with negative edge weights (denoted negative-weight SSSP) is one of the most fundamental and long-standing problems in graph algorithms, and has witnessed waves of exciting improvements every few decades since the 50s. Early works in the 50s, due to Shimbel~\cite{shimbel55}, Ford~\cite{ford56}, Bellman ~\cite{bellman58}, and Moore~\cite{moore59} resulted in the $O(mn)$ runtime. 
In the 80s and 90s, the scaling technique led to a wave of improvements (Gabow~\cite{Gabow85}, Gabow and Tarjan~\cite{GabowT89}, and Goldberg~\cite{Goldberg95}), resulting in runtime $O(m\sqrt{n}\log W)$, where $W\geq 2$ is the minimum integer such that $w(e)\geq -W$ for all $e\in E$.\footnote{
The case when $n$ is big and $W$ is small can be improved by the $O(n^\omega W)$-time algorithms of Sankowski~\cite{s05}, and Yuster and Zwick~\cite{yz05}.}
In the last few years, advances in continuous optimization and dynamic algorithms have led to a new wave of improvements, which achieve faster algorithms for the more general problems of transshipment and min-cost flow, and thus imply the same bounds for negative-weight SSSP (Cohen, Madry, Sankowski, Vladu  \cite{cmsv17}; Axiotis, Madry, Vladu  \cite{amv20}; BLNPSSSW \cite{BrandLNPSSSW20,BrandLLSS0W21-maxflow,blss20}). This line of work resulted in an near-linear runtime ($\tilde O((m+n^{1.5})\log W)$ time) on moderately dense graphs \cite{BrandLNPSSSW20} and $m^{4/3+o(1)}\log W$ runtime on sparse graphs \cite{amv20}.\footnote{$\Otil$-notation hides polylogarithmic factors. The dependencies on $W$ stated in \cite{amv20,BrandLNPSSSW20} are slightly higher than what we state here. These dependencies can be reduced by standard techniques (weight scaling, adding dummy source, and eliminating high-weight edges).}
For the special case of planar directed graphs~\cite{LiptonRT79,HenzingerKRS97,FakcharoenpholR06,KleinMW10,MozesW10}, near-linear time complexities were known since the 2001 breakthrough of Fakcharoenphol and Rao~\cite{FakcharoenpholR06} where the best current bound is $O(n\log^2(n)/\log\log n)$~\cite{MozesW10}. No near-linear time algorithm is known even for a somewhat larger class of graphs such as bounded-genus and minor-free graphs (which still requires $\tilde O(n^{4/3}\log W)$ time \cite{Wulff-Nilsen11}).
This state of the art motivates two natural questions: 
\begin{enumerate}
    \item \noindent {\em Can we get near-linear runtime for all graphs?} 
\item \noindent {\em Can we achieve efficient algorithms without complex machinery?} 
\end{enumerate}

For the second question, note that currently all state-of-the-art results for negative-weight SSSP are based on min-cost flow algorithms, and hence rely on sophisticated continuous optimization methods and a number of complex dynamic algebraic and graph algorithms (e.g. \cite{sw19,nsw17,cglnps20,BernsteinBNPSS20,ns17,w17}). It would be useful to develop  simple efficient algorithms that are specifically tailored to negative-weight SSSP, and thus circumvent the complexity currently inherent in flow algorithms; the best known bound of this kind is still the classic $O(m\sqrt{n}\log(W))$ from over three decades ago \cite{GabowT89,Goldberg95}. A related question is whether it is possible to achieve efficient algorithms for the problem using combinatorial tools, or whether there are fundamental barriers that make continuous optimization necessary.

\subsection{Our Result}
In this paper we resolve both of the above questions for negative-weight SSSP: we present a simple combinatorial algorithm that reduces the running time all the way down to near-linear.

\begin{theorem}
\label{thm:main-REALLY}
\label{thm:result}
There exists a randomized (Las Vegas) algorithm that takes $O(m\log^8(n)\log(W))$ time with high probability (and in expectation) for an $m$-edge input graph $\Gin$ and source $\sin$. It either returns a shortest path tree from $\sin$ or returns a negative-weight cycle. 
\end{theorem}

Our algorithm relies only on basic combinatorial tools; the presentation is self-contained and only uses standard black-boxes such as Dijkstra's and Bellman-Ford algorithms. In particular, it is a scaling algorithm enhanced by a simple graph decomposition algorithm called {\em Low Diameter Decomposition} which has been studied since the 80s; our decomposition is obtained in a manner similar to some known algorithms (see Section \ref{sec:decomposition} for a more detailed discussion). Our main technical contribution is showing how low-diameter decomposition---which works only on graphs with non-negative weights---can be used to develop a recursive scaling algorithm for SSSP with negative weights.
As far as we know, all previous applications of this decomposition were used for  parallel/distributed/dynamic settings for problems that do not involve negative weights, and our algorithm is also the first to take advantage of it in the classical sequential setting; we also show that in this setting, there is a simple and efficient algorithm to compute it.

\paragraph{Perspective on Other Problems:} While our result is specific to negative-weight SSSP, we note that question \emph{(2)} above in fact applies to a much wider range of problems. The current landscape of graph algorithms is that for many of the most fundamental problems, including ones taught in undergraduate courses and used regularly in practice, the state-of-the-art solution is a complex algorithm for the more general min-cost flow problem: some examples include negative-weight SSSP, bipartite matching, the assignment problem, edge/vertex-disjoint paths, s-t cut, densest subgraph, max flow, transshipment, and vertex connectivity. This suggests a research agenda of designing simple algorithms for these fundamental problems, and perhaps eventually their generalizations such as min-cost flow. We view our result on negative-weight SSSP as a first step in this direction.

\paragraph{Independent and Follow-Up Results.} Independently from our result, the recent major breakthrough by Chen, Kyng, Liu,  Peng, Probst Gutenberg,
and Sachdeva~\cite{ChenKLPGS22} culminates the line of works based on continuous optimization and dynamic algorithms (e.g. \cite{DaitchS08,Madry13,LeeS14,Madry16,cmsv17,cls19,b20,LiuS20_stoc,amv20,blss20,BrandLNPSSSW20,BrandLLSS0W21-maxflow}) and achieves an almost-linear time bound\footnote{$\tilde O(m^{1+o(1)}\log^2 U)$ time when vertex demands, edge costs, and upper/lower edge capacities are all integral and bounded by $U$ in absolute value.} 
for min-cost flow. 
The authors thus almost match our bounds for negative-weight SSSP as a special case of their result: their runtime is $m^{1+o(1)}\log(W)$ versus our $O(m\cdot \mathrm{polylog}(n)\log(W))$ bound. The two results are entirely different, and as far as we know there is no overlap in techniques. 

The above landmark result essentially resolves the running-time complexity for a wide range of fundamental graph problems, modulo the extra $m^{o(1)}$ factor. We believe that this makes it a natural time to pursue question {\em (2)} for these problems, outlined above.

In this paper, we prioritize simplicity and modularity, and not optimizing the logarithmic factors. The follow-up work by Bringmann, Cassis, and Fischer \cite{BringmannCF23} greatly optimizes our framework to reduce the running time to $O(m\log^2(n)$ $\log(nW)\log\log n).$ There is also follow-up work by Ashvinkumar \emph{et al.} showing how the framework can be applied to the parallel and distributed models of computation \cite{Ashvinkumar2023}.

\subsection{Techniques}\label{sec:decomposition}
Our main contribution is a new recursive scaling algorithm called $\ScaleDown$: see Section \ref{sec:ScaleDown}, including an overview in Section \ref{sec:ScaleDown:overview}. In this subsection, we highlight other techniques that may be of independent interest.

\paragraph{Low-Diameter Decomposition.} One of our key subroutines is an algorithm that decomposes any directed graph with \textit{non-negative} edge weights into strongly-connected components (SCCs) of small diameter. In particular, the algorithm computes a small set of edges $\esep$ such that all SCCs in the graph $G \setminus \esep$ have small weak diameter. Although the lemma below only applies to non-negative weights, we will show that it is in fact extremely useful for our problem.
\begin{restatable}{lemma}{SCCLemma}\label{lem:SCCDecomposition}
There is an algorithm $\LowDiamDecomposition(G,D)$ with the following guarantees:
\begin{itemize}[noitemsep]
    \item INPUT: an $m$-edge, $n$-vertex graph $G = (V,E,w)$ with non-negative integer edge weight function $w$ and a positive integer $D$.
    \item OUTPUT: A set of edges $\esep$ with the following guarantees:
    \begin{itemize}[noitemsep]
        \item each SCC of $G \setminus \esep$ has weak diameter at most $D$; that is, if $u,v$ are in the same SCC, then $\dist_G(u,v) \leq D$ and $\dist_G(v,u) \leq D$.
        \item For every $e \in E$, $\Pr[e \in \esep] = O\left(\frac{w(e) \cdot \log^2 n}{D} + n^{-10}\right)$. These probabilities are not guaranteed to be independent.\footnote{The $10$ in the exponent suffices for our application but can be replaced by an arbitrarily large constant.}
    \end{itemize}
    \item RUNNING TIME: The algorithm has running time $O(m\log^2 n+n\log^3 n)$.
\end{itemize}
\end{restatable}

The decomposition above is similar to other low-diameter decompositions used in both undirected and directed graphs, though the precise guarantees vary a lot between papers~\cite{Awerbuch85,AwerbuchGLP89,AwerbuchP92,AwerbuchBCP92,LinialS93,Bartal96, BlellochGKMPT14, MillerPX13, PachockiRSTW18, ForsterG19, ChechikZ20, BernsteinGW20,ForsterGV21}. The closest similarity is to the algorithm {\sc Partition} of Bernstein, Probst-Gutenberg, and Wulff-Nilsen \cite{BernsteinGW20}. The main difference is that the algorithm of \cite{BernsteinGW20} needed to work in a dynamic setting, and as a result their algorithm is too slow for our purposes. Our decomposition algorithm follows the general framework of \cite{BernsteinGW20}, but with several key differences to ensure faster running time; our algorithm is also simpler, since it only applies to the static setting. For the reader's convenience, we present the entire algorithm from scratch in Section \ref{sec:low-diam-full}.

\paragraph{No Negative-Weight Cycle Assumption via a Black-Box Reduction.}  Although it is possible to prove \Cref{thm:result} directly, the need to return the negative-weight cycle somewhat complicates the details. For this reason, we focus most of our paper on designing an algorithm that returns correct distances when the input graph contains no negative-weight cycle and guarantees nothing otherwise. 
See the description of subroutine $\SPmain$ in \Cref{thm:spmain} as an example. 
We can focus on the above case because we have a black-box reduction from \Cref{thm:main-REALLY} to the above case that incurs an extra $O(\log^2(n))$ factor in the runtime. 
This reduction is essentially the same as the reduction from the minimum cost-to-profit ratio cycle problem to the negative cycle detection problem \cite{Lawler66,Lawler76}.
We include our reduction in  \Cref{sec:las-vegas} for completeness. Note that the more general cost-to-profit ratio cycle problem can be solved in near-linear time via a similar reduction  \cite{Lawler66,Lawler76}.

%% file: prelim.tex
\section{Preliminaries}\label{sec:prelim}
Throughout, we only consider graphs with integer weights. 
For any weighted graph $G=(V,E, w)$, define $V(G) = V, E(G) = E$, and \[E^{neg}(G) := \{e \in E \mid w(e) < 0\} .\]
Define $W_G := \max\{2,-\min_{e \in E}\{w(e)\}\}$; that is, $W_G$ is the most negative edge weight in the graph\footnote{We set $W_G \geq 2$ so that we can write $\log(W_G)$ in our runtime}. Given any set of edges $S \subseteq E$ we define $w(S) = \sum_{e \in S} w(e)$. We say that a cycle $C$ in $G$ is a negative-weight cycle if $w(C) < 0$. We define $\dist_G(u,v)$ to be the shortest distance from $u$ to $v$; if there is a negative-weight cycle on some $uv$-path then we define $\dist_G(u,v) = -\infty$.

Consider graph $G = (V,E,w)$ and consider subsets $V' \subseteq V$ and $E' \subseteq E$. 
We define $G[V']$ to be the subgraph of $G$ induced by $V'$.
We slightly abuse notation and denote a subgraph of $G$ as $H = (V',E',w)$ where the weight function $w$ is restricted to edges in $H$. 
We define $G\setminus V'=G[V \setminus V']$ and $G\setminus E'=(V, E\setminus E', w)$; i.e. they are graphs where we remove vertices and edges in $V'$ and $E'$ respectively. 
We sometimes write $G\setminus v$ and $E\setminus e$ instead of $G\setminus \{v\}$ and $G\setminus \{e\}$, respectively, for any $v\in V$ and $e\in E$. 
We say that a subgraph $H$ of $G$ has \emph{weak} diameter $D$ if for any $u,v \in V(H)$ we have that $\dist_G(u,v) \leq D$.
We always let $\gin$ and $\sin$ refer to the main input graph/source of \Cref{thm:main-REALLY}.

\begin{assumption}[Properties of input graph $\gin$; justified by Lemma \ref{lem:assumption}]
\label{assum:main}
\label{assum:small-weights}
We assume throughout the paper that the main input graph $\gin = (V,E,\win)$ satisfies the following properties:
\begin{enumerate}
    \item $\win(e) \geq -1$ for all $e \in E$ (thus, $W_{\gin} = 2$).
    \item Every vertex in $\gin$ has constant out-degree.
\end{enumerate}
\end{assumption}

\begin{lemma} 
\label{lem:assumption}
Say that there is an algorithm as in \Cref{thm:result} for the special case when the graph $\gin$ satisfies the properties of Assumption \ref{assum:main}, with running time $T(m,n)$. Then there is algorithm as in  \Cref{thm:result} for \textit{any} input graph $\gin$ with integral weights that has running time $O(T(m,m)\log({W_{\gin}}))$. 
\end{lemma}

\begin{proof}
Let us first consider the first assumption, i.e that $\win(e) \geq -1$. The scaling framework of Goldberg~\cite{Goldberg95} shows that an algorithm for this case implies an algorithm for any integer-weighted $G$ at the expense of an extra $\log(W_G)$ factor.\footnote{Quoting \cite{Goldberg95}: ``Note that the basic problem solved at each iteration of the bit scaling method is a special version of the shortest paths problem where the arc lengths are integers greater or equal to $-1$.''} 

For the assumption that every vertex in $\gin$ has constant out-degree, we use a by-now standard technique of creating $\Theta$(out-degree($v$)) copies of each vertex $v$, so that each copy has constant out-degree; the resulting graph was $O(E)$ vertices and $O(E)$ edges.\footnote{One way to do this is to replace every vertex with a directed zero-weight cycle whose size is the in-degree plus out-degree of the vertex and then attach the adjacent edges to this cycle.}
\end{proof}

\paragraph{Dummy Source and Negative Edges.}
The definitions below capture a common transformation we apply to negative weights and also allow us to formalize the number of negative edges on a shortest path. Note that most of our algorithms/definitions will not refer to the input source $\sin$, but instead to a dummy source that has edges of weight $0$ to every vertex. 

\begin{definition}[$G_s$,$w_s$,$G^B$,$w^B$,$G^B_s$,$w^B_s$]
\label{dfn:dummy} \label{def:dummy}
\label{def:GB}
Given any graph $G = (V,E,w)$, we let $G_s=(V\cup \{s\}, E\cup \{(s,v)\}_{v\in V}, w_s)$ refer to the graph $G$ with a dummy source $s$ added, where there is an edge of weight $0$ from $s$ to $v$ for every $v \in V$ and no edges into $s$. Note that $G_s$ has a negative-weight cycle if and only if $G$ does and that $\dist_{G_s}(s,v) = \min_{u \in V} \dist_G(u,v)$.

For any integer $B$, let $G^B=(V, E, w^B)$ denote the graph obtained by adding $B$ to all negative edge weights in $G$, i.e. $w^B(e) = w(e) + B$ for all $e \in \eneg(G)$ and $w^B(e) = w(e)$ for $e \in E \setminus \eneg(G)$. 
Note that $(G^B)_s=(G_s)^B$ so we can simply write $G^B_s=(V\cup \{s\}, E\cup \{(s,v)\}_{v\in V}, w^B_s)$. 
\end{definition}

\begin{definition}
[$\eta_G(v)$, $P_G(v)$]\label{def:eta}\label{def:prelim:the shortest path}\label{def:prelim:delta negative}\label{def:prelim:good vertex}
For any graph $G = (V,E,w)$ and $s \in V$ such that $s$ can reach all vertices in $V$, define 
\[\eta_G(v; s):=
\begin{cases}
\infty & \mbox{if $\dist_{G}(s,v)=-\infty$} \\
\min\{|\eneg(G) \bigcap P| : \mbox{$P$ is a shortest $sv$-path in $G$}\}; & \mbox{otherwise.}
\end{cases}
\]
Let $\eta(G;s)=\max_{v\in V} \eta_G(v; s)$. 
When $\dist_G(s,v)\neq -\infty$, let $P_G(v; s)$ be a shortest $sv$-path on $G$ such that 
\begin{align}
|\eneg(G) \bigcap P_G(v)|=\eta_G(v). \label{eq:def:P_G(v)} 
\end{align}
We omit $s$ when $s$ is the dummy source in $G_s$ defined in \Cref{def:dummy}; i.e. $\eta_G(v)=\eta_{G_s}(v; s)$, $\eta(G)=\eta(G_s;s)$, and $P_G(v)=P_{G_s}(v; s).$ 
Further, the graph $G$ is omitted if the context is clear.  
\end{definition}

\subsection{Price Functions and Equivalence}

Our algorithm heavily relies on price functions---an idea with deep roots, traceable to Dreyfus' 1969 survey \cite{Dreyfus69} and Braess’ work \cite[Page 176]{Braess71}, if not earlier---and which also appeared in the well-known $A^*$ algorithm (referred to as a heuristic function in the AI community).\footnote{We thank Donald E. Knuth for the historical note regarding price functions.}

\begin{definition}[Price Function]\label{def:prelim:price}
Consider a graph $G = (V,E,w)$ and let $\phi$ be any function: $V \rightarrow \mathbb{Z}$, where $\mathbb{Z}$ is the set of integers. Then, we define $w_\phi$ to be the weight function $w_\phi(u,v) = w(u,v) + \phi(u) - \phi(v)$ and we define $G_\phi = (V,E,w_\phi)$. We will refer to $\phi$ as a \emph{price} function on $V$. Note that $(G_{\phi})_{\psi} = G_{\phi + \psi}$. 
\end{definition}

\begin{definition}[Graph Equivalence]
\label{def:equivalent}
We say that two graphs $G = (V,E,w)$ and $G' = (V,E,w')$ are \emph{equivalent} if {\bf (1)} any shortest path in $G$ is also a shortest path in $G'$ and vice-versa and {\bf (2)} $G$ contains a negative-weight cycle if and only if $G'$ does.
\end{definition}

\begin{lemma}[\cite{Johnson77}]
\label{lem:price-equivalent}
Consider any graph $G = (V,E,w)$ and price function $\phi$. For any pair $u, v \in V$ we have $\dist_{G_\phi}(u,v) = \dist_G(u,v) + \phi(u) - \phi(v)$, and for any cycle $C$ we have $w(C) = w_{\phi}(C)$. As a result, $G$ and $G_{\phi}$ are equivalent. Finally, if $G' = (V,E,w)$, $G' = (V,E,w')$ and $w' = c \cdot w(e)$ for some positive $c$, then $G$ and $G'$ are equivalent. 
\end{lemma}

The overall goal of our algorithm will be to compute a price function $\phi$ such that all edge weights in $G_{\phi}$ are non-negative (assuming no negative-weight cycle); we can then run Dijkstra on $G_{\phi}$. The lemma below, originally used by Johnson, will be one of the tools we use.

\begin{lemma}[\cite{Johnson77}]
\label{lem:price-source}\label{lem:preliim:price=distance}
Let $G = (V,E)$ be a directed graph with no negative-weight cycle and let $s$ be any vertex in $G$ that can reach all vertices in $V$. Let $\phi(v) = \dist_{G}(s,v)$ for all $v \in V$. Then, all edge weights in $G_\phi$ are non-negative. (The lemma follows trivially from the fact that $\dist(s,v) \leq \dist(s,u) + w(u,v)$.)
\end{lemma}

%% file: sec_algorithms.tex
\section{The Framework}
\label{sec:algorithm}
\label{sec:framework}

In this section we describe the input/output guarantees of all the subroutines used in our algorithm, as well as some of the algorithms themselves.

\subsection{Basic Subroutines}

\begin{lemma}[$\Dijkstra$]\label{thm:prelim:Dijkstra}
There exists an algorithm $\Dijkstra(G,s)$ that takes as input a graph $G$ with non-negative edge weights and a vertex $s \in V$ and outputs a shortest path tree from $s$ in $G$. The running time is $O(m + n\log(n))$. 
\end{lemma}

It is easy to see that if $G$ is a DAG (Directed Acyclic Graph), computing a price function $\phi$ such that $G_{\phi}$ has non-negative edge weights is straightforward: simply loop over the topological order $v_1, ..., v_n$ and set $\phi(v_i)$ so that all incoming edges have non-negative weight.  
The lemma below generalizes this approach to graphs where only the ``DAG part'' has negative edges.

\begin{lemma}[$\FixAlmostDag$]\label{lem:almostDag}
There exists an algorithm $\FixAlmostDag(G, \cal P)$ that  takes as input a graph $G$ and a partition ${\cal P}:=\{V_1, V_2, \ldots\}$ of vertices of $G$ such that 
\begin{enumerate}[noitemsep,nolistsep]
    \item for every $i$, the induced subgraph $G[V_i]$ contains no negative-weight edges, and
    \item when we contract every $V_i$ into a node, the resulting graph is a DAG (i.e. contains no cycle). 
\end{enumerate}
The algorithm outputs a price function $\phi:V\rightarrow \mathbb{Z}$ such that $w_{\phi}(u,v) \geq 0$ for all $(u,v) \in E$.  
The running time is $O(m + n)$.
\end{lemma}

\begin{proof}[Proof sketch]
The algorithm is extremely simple: it loops over the SCCs $V_i$ in topological order, and when it reaches $V_i$ it sets the same price $\phi(v)$ for every $v \in V_i$ that ensures there are no non-negative edges entering $V_i$; since all $\phi(v)$ are the same, this does not affect edge-weights inside $V_i$. We leave the pseudocode and analysis for Section \ref{sec:app:fixdag} in the appendix.
\end{proof}

The next subroutine shows that computing shortest paths in a graph $G$ can be done efficiently as long as $\eta(v)$ is small on average (see Definition \ref{def:prelim:delta negative} for $\eta(v)$). Note that this subroutine is the reason we use the assumption that every vertex has constant out-degree (\Cref{assum:main}).

\begin{restatable}{lemma}{spaverageLemma}{\em ($\SPaverage$)}\label{lem:spaverage}
There exists an algorithm $\SPaverage(G,s)$ that takes as input a graph $G = (V,E,w)$ in which all vertices $v \neq s$ have constant out-degree and a source $s\in V$ that can reach all vertices in $V$. The algorithm outputs a price function $\phi$ such that $w_\phi(e) \geq 0$ for all $e \in E$ and has running time $O(\log(n) \cdot (n + \sum_{v \in V} \eta_G(v; s)))$ (Definition \ref{def:prelim:delta negative}); note that if  $G$ contains a negative-weight cycle then  $\sum_{v \in V} \eta_G(v; s) = \infty$ so the algorithm will never terminate and hence not produce any output.
\end{restatable}
\begin{proof}[Proof sketch]
By Lemma \ref{lem:price-equivalent}, in order to compute the desire price function, it suffices to describe how $\SPaverage(G)$ computes $\dist_{G}(s,v)$ for all $v\in V$, where $s$ is the input source to $\SPaverage(G,s)$.

The algorithm is a straightforward combination of Dijkstra's and Bellman-Ford's algorithms. The algorithm maintains distance estimates $d(v)$ for each vertex $v$. It then proceeds in multiple iterations, where each iteration first runs a Dijkstra Phase that ensures that all non-negative edges are relaxed and then a Bellman-Ford Phase ensuring that all negative edges are relaxed. Consider a vertex $v$ and let $P$ be a shortest path from $s$ to $v$ in $G$ with $\eta_G(v; s)$ edges of $\eneg(G)$. It is easy to see that $\eta_G(v;s) + 1$ iterations suffice to ensure that $d(v) = \dist_{G}(s,v)$. In each of these iterations, $v$ is extracted from and added to the priority queue of the Dijkstra Phase only $O(1)$ times. Furthermore, after the $\eta_G(v;s) + 1$ iterations, $v$ will not be involved in any priority queue operations. Since the bottleneck in the running time is the queue updates, we get the desired time bound of $O(\log(n)\cdot\sum_{v\in V}(\eta_G(v; s) + 1)) = O(\log(n)\cdot (n + \sum_{v\in V}\eta_G(v; s)))$.

This completes the proof sketch. The full proof can be found in Appendix~\ref{sec:app:spaverage}.
\end{proof}

\subsection{The Interface of the Two Main Algorithms}
\label{sec:main-algos}

Our two main algorithms are called $\ScaleDown$ and $\SPmain$. The latter is a relatively simple outer shell. The main technical complexity lies in $\ScaleDown$, which calls itself recursively.

\begin{theorem}[$\SPmain$]
\label{thm:spmain}
There exists an algorithm $\SPmain(\gin,\sin)$ that takes as input a graph $\gin$ and a source $\sin$ satisfying the properties of Assumption \ref{assum:main}. If the algorithm terminates, it outputs a shortest path tree $T$ from $\sin$. The running time guarantees
are as follows:
\begin{itemize}
    \item If the graph $\gin$ contains a negative-weight cycle then the algorithm never terminates.
    \item If the graph $\gin$ does not contain a negative-weight cycle then the algorithm has \textit{expected} running time $\tspmain = O(m\log^5(n))$. 
\end{itemize}
\end{theorem}

\begin{theorem}[$\ScaleDown$]
\label{thm:scaledown}
There exists the following algorithm $\ScaleDown(G = (V,E,w),\Delta,B)$.
\begin{enumerate}
    \item INPUT REQUIREMENTS: \label{item:ScaleDown input}
    \begin{enumerate}
        \item\label{prop:scaledown:weight} $B$ is positive integer, $w$ is integral, and $w(e) \geq -2B$ for all $e \in E$
        \item\label{prop:scaledown:Delta} 
       If the graph $G$ does not contain a negative-weight cycle then the input must satisfy $\eta(G^{B})\leq \Delta$; that is, for every $v \in V$  
        there is a shortest $sv$-path in $G^{B}_s$ with at most $\Delta$ negative edges
        (\Cref{def:dummy,def:prelim:delta negative})\label{item:ScaleDown input b}
        \item\label{prop:scaledown:degree} All vertices in $G$ have constant out-degree
    \end{enumerate}
    \item \label{item:ScaleDown Output} OUTPUT: If it terminates, the algorithm returns an integral price function $\phi$ such that $w_{\phi}(e) \geq -B$ for all $e \in E$
    \item RUNNING TIME:\label{item:ScaleDown Runtime}  If $G$ does not contain a negative-weight cycle, then the algorithm has expected runtime $O\left(m\log^3(n)\log(\Delta)\right)$. Remark: If $G$ contains a negative-weight cycle, there is no guarantee on the runtime, and the algorithm might not even terminate; but if the algorithm does terminate, it always produces a correct output.
\end{enumerate}
\end{theorem}

\paragraph{Remark: Termination and Negative-Weight Cycles.}
Note that for both algorithms above, if $\Gin$ contains a negative-weight cycle then the algorithm might simply run forever, i.e. not terminate and not produce any output. In fact the algorithm $\SPmain$ \emph{never} terminates if $\Gin$ contains a negative-weight cycle. The algorithm $\ScaleDown$ may or may not terminate in this case: our guarantee is only that if it does terminate, it always produces a correct output.

In short, neither algorithm is required to produce an output in the case where $\Gin$ contains a negative-weight cycle, so we recommend the reader to focus on the case where $\Gin$ does not contain a negative-weight cycle.

\paragraph{Proof sketch of Theorem \ref{thm:result}.}
Algorithm $\SPmain$ is almost what we need for the main result in Theorem \ref{thm:result}, but the guarantees of $\SPmain$ are slightly weaker: it is Monte Carlo (rather than Las Vegas), and it cannot return a negative cycle. We show in Section \ref{sec:las-vegas} that a combination of simple probabilistic bootstrapping and binary search can be used to transform $\SPmain$ into the main result promised in Theorem \ref{thm:result}.

%% file: sec_scaledown.tex
\section{Algorithm $\ScaleDown$ (\Cref{thm:scaledown})}\label{sec:ScaleDown}

We start by describing the algorithm $\ScaleDown$, as this contains our main conceptual contributions; the much simpler algorithm $\SPmain$ is described in the following section. 
Full pseudocode of $\ScaleDown$ is given in \Cref{alg:ScaleDown}. The algorithm mostly works with graph $G^B=(V, E, \wB)$.
For the analysis, the readers may want to familiarize themselves with, e.g., $G^B_s$, $w^B_s$, $P_{G^B}(v)$ and $\eta(G^B)$ from \Cref{def:eta,def:dummy}. In particular, throughout this section,  source $s$ \emph{always} refers to a dummy source that has outgoing edges to every vertex in $V$ and has no incoming edges.

Note that 
$m=\Theta(n)$ 
since the input condition requires constant out-degree for every vertex. So, we use $m$ and $n$ interchangeably in this section. We briefly describe the $\ScaleDown$ algorithm and sketch the main ideas of the analysis in \Cref{sec:ScaleDown:overview}, before showing the full analysis in \Cref{sec:ScaleDown:Analysis}.

\input{ScaleDown_pseudo}

\subsection{Overview}\label{sec:ScaleDown:overview}

The algorithm runs in phases, where in the last phase it calls $\SPaverage((\GB_s)_{\phi_2}, s)$ for some price function $\phi_2$ (\Cref{line:ScaleDown:Phase 3}). Recall (\Cref{lem:spaverage}) that if $\SPaverage$ terminates, it returns price function $\psi'$ 
such that all edge weights in $(\GB_s)_{\phi_2+\psi'}$ are non-negative. Consequently, all edge weights in $\GB_{\phi_3}$ are non-negative for $\phi_3=\phi_2+\psi'$ (since $\GB_{\phi_2}$ is a subgraph of  $(\GB_s)_{\phi_2}$). 
We thus have $w_{{\phi_3}}(e)\geq -B$ for all $e\in E$ as desired (because $\wB(e)\leq w(e)+B$). 
This already proves the output correctness of $\ScaleDown$ (\Cref{item:ScaleDown Output} of \Cref{thm:scaledown}). (See \Cref{thm:ScaleDown:main output correctness} for the detailed proof.) 
Thus it remains to bound the runtime when $G$ contains no negative-weight cycle (\Cref{item:ScaleDown Runtime}).
{\em In the rest of this subsection we assume that $G$ contains no negative-weight cycle.}

Bounding the runtime when $\Delta\leq 2$ is easy: The algorithm simply jumps to Phase 3 with $\phi_2=0$ (\Cref{line:ScaleDown:BaseCase} in \Cref{alg:ScaleDown}). 
Since $\eta(\GB)\leq \Delta\leq 2$ (the input requirement; \Cref{item:ScaleDown input b}),  the runtime of $\SPaverage((\GB_s)_{\phi_2}, s)$ is $O((m+\sum_{v\in V} \eta_{\GB}(v))\log m)=O(m\Delta\log m)=O(m\log m)$.

For $\Delta>2$, we require some properties from Phases 0-2 in order to bound the runtime of Phase~3. In Phase 0, we partition vertices into strongly-connected components (SCCs)\footnote{Recall that a SCC is a {\em maximal} set $C\subseteq V$ such that for every $u,v\in V$, there are paths from $u$ to $v$ and from $v$ to $u$. See, e.g., Chapter 22.5 in \cite{CLRS_book_3rd}.} $V_1, V_2, \ldots$ such that each $V_i$ has weak diameter $dB=B\Delta/2$ in $G$. We do this by calling $\esep \gets \SCCDecomposition(\GB_{\geq 0}, dB)$, where $\GB_{\geq 0}$ is obtained by rounding all negative weights in $\GB$ up to $0$; we then let $V_1, V_2, \ldots$ be the SCCs of $\GB \setminus \esep$. (We need $\GB_{\geq 0}$ since $\SCCDecomposition$ can not handle negative weights.\footnote{One can also use $G_{\geq 0}$ instead of $\GB_{\geq 0}$. We choose $\GB_{\geq 0}$ since some proofs become slightly simpler.})  See \Cref{thm:ScaleDown:Weak Diameter} for the formal statement and proof.

The algorithm now proceeds in three phases. In Phase 1 it computes a price function $\phi_1$ that makes the edges inside each SCC $V_i$ non-negative; in Phase 2 it computes $\phi_2$ such that the edges between SCCs in $\GB \setminus \esep$ are also non-negative; finally, in Phase 3 it makes non-negative the edges in $\esep$ by calling $\SPaverage$.

{\bf (Phase 1)} 
Our goal in Phase 1 is to compute $\phi_1$ such that $\wB_{\phi_1}(e)\geq 0$ for every edge $e$ in $\GB[V_i]$ for all $i$.
To do this, we recursively call $\ScaleDown(H, \Delta/2, B)$, where $H$ is a union of all the SCCs $G[V_i]$. 
The main reason that we can recursively call  $\ScaleDown$ with parameter $\Delta/2$ is because we can argue that, 
when $G$ does not contain a negative-weight cycle,  
$$\eta(\HB)\leq d=\Delta/2.$$ 
As a rough sketch, the above bound holds because if any shortest path $P$ from dummy source $s$ in some $(\GB[V_i])_s$ contains more than $d$ negative-weight edges, then it can be shown that $w(P)<-dB$; this is the step where we crucially rely on the difference between $w^B(P)$ and $w(P)$. Combining $w(P)<-dB$ with the fact that $\GB[V_i]$ has weak diameter at most $dB$ implies that $G$ contains a negative-weight cycle.
See \Cref{thm:phase 1 works} for the detailed proof.

{\bf (Phase 2)} Now that all edges in $\GB_{\phi_1}[V_i]$ are non-negative, we turn to the remaining edges in $\GB\setminus \esep$.
Since these remaining edges (i.e. those not in the SCCs) form a directed acyclic graph (DAG), we can simply call  $\FixAlmostDag(\GB_{\phi_1} \setminus \esep, \{V_1, V_2, \ldots \})$ (\Cref{lem:almostDag}) to get a price function $\psi$ such that all edges in $(\GB_{\phi_1} \setminus \esep)_{\psi}=\GB_{\phi_2} \setminus \esep$ are non-negative. (See \Cref{thm:ScaleDown:Phase 2 conclusion}.)

{\bf (Phase 3)} 
By the time we reach this phase, the only negative edges remaining in $\GB_{\phi_2}$ are the ones in $\esep$; that is, $\eneg(\GB_{\phi_2}) \subseteq \esep$. We call  $\SPaverage((\GB_s)_{\phi_2}, s)$ to eliminate these negative edges.
We are now ready to
show that the runtime of Phase 3, which is $O((m+\sum_{v\in V} \eta_{(\GB_s)_{\phi_2}}(v;s))\log m)$ (Lemma \ref{lem:spaverage}), is $O(m\log^3 m)$ in expectation.
We do so by proving that for any $v\in V$,
\[E\left[\eta_{(\GB_s)_{\phi_2}}(v; s)\right]= O(\log^2 m).\] 
(See Equation \eqref{eq:Phase 3 expectation} near the end of the next subsection.)
A competitive reader might want to try to prove the above via a series of inequalities: 
$\eta_{(\GB_s)_{\phi_2}}(v)\leq |P_{\GB}(v)\cap \eneg(\GB_{\phi_2})|+1 \leq |P_{\GB}(v)\cap \esep|+1$, and also, the guarantees of $\SCCDecomposition$ (Lemma \ref{lem:SCCDecomposition}) imply that after Phase 0, $E\left[|P_{\GB}(v)\cap \esep|\right]=O(\log^2 m).$ (Proved in \Cref{thm:ScaleDown:expected esep}.)

Finally, observe that there are $O(\log \Delta)$ recursive calls, and the runtime of each call is dominated by the $O(m\log^3 m)$ time of Phase 3. So, the total expected runtime is $O(m\log^3(m)\log \Delta)$

\paragraph{Remark.}
Our sequence of phases 0-3 is reminiscent of the sequencing used by Bernstein, Probst-Gutenberg, and Saranurak in their result on dynamic reachability \cite{BernsteinGS20}, although the actual work within each phase is entirely different, and the decompositions have different guarantees.  The authors of \cite{BernsteinGS20} decompose the graph into a DAG of {\em expanders} plus some separator edges (analogous to our phase 0); they then handle reachability inside expanders (phase 1), followed by reachability using the DAG edges (phase 2), and finally incorporate the separator edges (phase 3).

\subsection{Full Analysis}\label{sec:ScaleDown:Analysis}

\Cref{thm:scaledown} follows from 
\Cref{thm:ScaleDown:main output correctness,thm:ScaleDown:main runtime} below. We start with \Cref{thm:ScaleDown:main output correctness} which is quite trivial to prove.

\begin{theorem}\label{thm:ScaleDown:main output correctness}
$\ScaleDown(G=(V,E,w), \Delta,B)$ either does not terminate or returns $\phi=\phi_3$ such that $w_\phi(e)\geq -B$ for all $e\in E$. 
\end{theorem}
\begin{proof}
Consider when we call $\SPaverage((\GB_s)_{\phi_2}, s)$(\Cref{lem:spaverage}) in Phase 3 for some integral price function $\phi_2.$ Either this step does not terminate or returns an integral price function $\psi'$ such that
$(\GB_{\phi_2})_{\psi'}=\GB_{\phi_2+\psi'}=\GB_{\phi_3}$ contains no negative-weight edges; i.e. $\wB_{\phi_3}(e)\geq 0$ for all $e\in E$.  
Since $\wB(e)\leq w(e)+B$, we have $w_{\phi_3}(e)\geq \wB_{\phi_3}(e)-B\geq -B$ for all $e\in E$. 
\end{proof}

\Cref{thm:ScaleDown:main output correctness} implies that the output condition of $\ScaleDown$ (\cref{item:ScaleDown Output} in \Cref{thm:scaledown}) is always satisfied, regardless of whether $G$ contains a negative-weight cycle or not. 
It remains to show that if $G$ does not contain a negative-weight cycle, then $\ScaleDown(G=(V,E,w), \Delta,B)$ has expected runtime of $O(m\log^3(m)\log(\Delta))$. It suffices to show the following.

\begin{theorem}\label{thm:ScaleDown:main runtime}
If $G$ does not contain a negative-weight cycle, then
the expected time complexity of Phase 3 is $O(m\log^3 m)$. 
\end{theorem}
This suffices because, first of all, it is easy to see that Phase 0 requires $O(m \log^3(m))$ time (by \Cref{lem:SCCDecomposition}) and other phases (except the recursion on  \Cref{line:ScaleDown:Recursion}) requires $O(m+n)$ time. Moreover, observe that if $G$ contains no negative-weight cycle, then the same holds for $H$ in the recursion call $\ScaleDown(H, \Delta/2, B)$ (\Cref{line:ScaleDown:Recursion} of \Cref{alg:ScaleDown}); thus, if $G$ contains no negative-weight cycle, then all recursive calls also get an input with no negative-weight cycle. 
So, by \Cref{thm:ScaleDown:main runtime} the time to execute a single call in the recursion tree is $O(m\log^3 m)$ in expectation. 
Since there are $O(\log \Delta)$ recursive calls, the total running time is $O(m\log^3(m)\log(\Delta))$ by linearity of expectation.

\paragraph{Proof of \Cref{thm:ScaleDown:main runtime}.}
The rest of this subsection is devoted to proving \Cref{thm:ScaleDown:main runtime}. From now on, we consider any graph $G$ that does not contain a negative-weight cycle. (We often continue to state this assumption in lemma statements so that they are self-contained.)

\paragraph{Base case: $\Delta\leq 2$.} This means that for every vertex $v$, $\eta_{\GB}(v)\leq \eta(\GB)\leq \Delta\leq 2$ (see the input requirement of $\ScaleDown$ in \Cref{item:ScaleDown input b} of \Cref{thm:scaledown}). So, the runtime of Phase 3 is $$O\left(\left(m+\sum_{v\in V} \eta_{\GB}(v)\right)\log m\right)=O\left(m\Delta\log m\right)=O\left(m\log m\right).$$

We now consider when $\Delta>2$ and show properties achieved in each phase. We will use these properties from earlier phases in analyzing the runtime of Phase 3.

\paragraph{Phase 0: Low-diameter Decomposition.} 
It is straightforward that the SCCs $G[V_i]$ have weak diameter at most $dB$ (this property will be used in Phase 1): 

\begin{lemma}\label{thm:ScaleDown:Weak Diameter}
For every $i$ and every $u,v\in V_i$, $\dist_G(u,v)\leq dB.$
\end{lemma}
\begin{proof}
For every $u,v\in V_i$, we have $\dist_G(u,v)\leq \dist_{\GB_{\geq 0}}(u, v)  \leq dB$ where the first inequality is because $w(e)\leq \wB_{\geq 0}(e)$ for every edge $e\in E$ and the second inequality is by the output guarantee of $\SCCDecomposition$ (\Cref{lem:SCCDecomposition}). 
\end{proof}

Another crucial property from the decomposition is this: Recall from \Cref{def:eta} that $P_{\GB}(v)$ is the shortest $sv$-path in $\GB_s$ with $\eta_{\GB}(v)$ negative-weight edges.
We show below that in expectation $P_{\GB}(v)$ contains only $O(\log^2 n)$ edges from $\esep$. This will be used in Phase 3.

\begin{lemma}\label{thm:ScaleDown:expected esep}
If $\eta(\GB)\leq \Delta$, then for every $v\in V$, $E\left[\left|P_{\GB}(v)\cap \esep\right|\right]=O(\log^2 m)$. 
\end{lemma}
\begin{proof} Consider any $v\in V$.  The crux of the proof is the following bound on the weight of $P_{\GB}(v)$ in $\GB_{\geq 0}$:
\begin{align}
    \wB_{\geq 0}(P_{\GB}(v))\leq \eta_{\GB}(v)\cdot B\label{eq:ScaleDown:expected esep}
\end{align}
where we define $\wB_{\geq 0}(s,u)=0$ for every $u\in V$.  Recall the definition of of $\wB_s$ from \Cref{def:dummy} and note that $\wB_s(P_{\GB}(v)) \leq 0$ because there is an edge of weight $0$ from $s$ to every $v \in V$. We thus have
\Cref{eq:ScaleDown:expected esep} because
\begin{align}
    \wB_{\geq 0}(P_{\GB}(v))&\leq \wB_s(P_{\GB}(v))+  |P_{\GB}(v)\cap \eneg(\GB)|
    \cdot B &\mbox{since $\wB(e)\geq -B$ for all $e\in E$}\nonumber\\
    &\leq |P_{\GB}(v)\cap \eneg(\GB)| \cdot B &\mbox{since $\wB_s(P_{\GB}(v))\leq 0$}\nonumber\\
    & = \eta_{\GB}(v) \cdot B &
    \mbox{by definition of $P_{\GB}(v)$} 
    \nonumber
\end{align}
Recall from the output guarantee of $\SCCDecomposition$ (\Cref{lem:SCCDecomposition}) that $\Pr[e \in \esep] = O(\wB_{\geq 0}(e) \cdot (\log n)^2 / D + n^{-10})$, where in our case $D=dB=B\Delta/2$. This, the linearity of expectation, and \eqref{eq:ScaleDown:expected esep}
imply that
\begin{align*}
  E[P_{\GB}(v)\cap \esep] & = O\left(\frac{\wB_{\geq 0}(P_{\GB}(v))\cdot(\log n)^2}{B\Delta/2}+|(P_{\GB}(v))|\cdot n^{-10}\right)\\
  &\stackrel{\eqref{eq:ScaleDown:expected esep}}{=} O\left(\frac{2\eta_{\GB}(v)\cdot (\log n)^2}{\Delta}+n^{-9}\right) 
\end{align*}
which is $O(\log^2 n)$ when $\eta(\GB)\leq \Delta$.  
\end{proof}

\paragraph{Phase 1: Make edges inside the SCCs $\GB[V_i]$ non-negative.} 
We argue that $\ScaleDown(H, \Delta/2, B)$ is called with an input that satisfies its input requirements  (\Cref{thm:scaledown}). The most important requirement is $\eta(\HB)\leq \Delta/2$ (\Cref{item:ScaleDown input b}) which we prove below (other requirements are trivially satisfied). Recall that we set $d := \Delta/2$ in Line \ref{line:ScaleDown:def}.

\begin{lemma}\label{thm:phase 1 works}
If $G$ has no negative-weight cycle, then $ \eta(H^{B})\leq d=\Delta/2$. 
\end{lemma}
\begin{proof}
Consider any vertex $v\in V$. Let
$P:=P_{H^{B}}(v)\setminus s;$
i.e. $P$ is obtained by removing $s$ from a shortest $sv$-path in $H_s^{B}$ that contains $\eta_{H^B}(v)$ negative weights in $H_s^B$.
Let $u$ be the first vertex in $P$. 
Note three easy facts: 
\begin{itemize}[noitemsep]
    \item[(a)] $w_{\HB}(e)= w_H(e)+B$ for all $e\in \eneg(\HB)$,
    \item[(b)] $|\eneg(\HB)\cap P|=|\eneg(\HB)\cap P_{H^{B}}(v)|=\eta_{H^B}(v)$, and
    \item[(c)] $w_{\HB}(P)= w_{\HB_{s}}(P_{H^{B}}(v))\leq w_{\HB_{s}}(s,v)=0,$
\end{itemize}
where (b) and (c) are because the edges from $s$ to $u$ and $v$ in $H^{B}_s$ have weight zero. 
Then,\footnote{The ``$\stackrel{(a)}{\leq}$'' part in \eqref{eq:ScaleDown:Phase 1 two} is not equality because there can be an edge in $\eneg(H)\cap P$ that is not in $\eneg(\HB)\cap P.$ } 
\begin{align}
    \dist_G(u, v) &\leq w_{H}(P) \nonumber
    \stackrel{(a)}{\leq} w_{H^{B}}(P)-|\eneg(\HB)\cap P|\cdot B\nonumber\\
    %
    &\stackrel{(b)}{=} w_{H^{B}}(P)-\eta_{H^{B}}(v)\cdot B
    \stackrel{(c)}{\leq} -\eta_{H^{B}}(v)\cdot B.\label{eq:ScaleDown:Phase 1 two}
\end{align} 

Note that $u$ and $v$ are in the same SCC $V_i$;\footnote{in fact all vertices in $P$ are in the same SCC $V_i$, because we define $H=\bigcup_i G[V_i]$.} thus, by \Cref{thm:ScaleDown:Weak Diameter}: 
\begin{align}
\dist_G(v, u)\leq dB. \label{eq:ScaleDown:Phase 1 one}
\end{align}

If $G$ contains no negative-weight cycle, then $\dist_G(u,v)+\dist_G(v,u)\geq 0$ and thus $\eta_{H^{B}}(v)\leq dB\cdot (1/B)=d$ by \Cref{eq:ScaleDown:Phase 1 one,eq:ScaleDown:Phase 1 two}. Since this holds for every $v\in V$, \Cref{thm:phase 1 works} follows. 
\end{proof}

Consequently, $\ScaleDown$ (\Cref{thm:scaledown}) is guaranteed to output $\phi_1$ as follows.
\begin{corollary}\label{thm:ScaleDown:Phase 1 conclusion}
If $G$ has no negative-weight cycle, then all edges in $\GB_{\phi_1}[V_i]$ are non-negative for every $i$. 
\end{corollary}

\paragraph{Phase 2:  Make all edges in $\GB \setminus \esep$ non-negative.} Now that all edges in $\GB_{\phi_1}[V_i]$ are non-negative, we turn to the remaining edges in $\GB\setminus \esep$.
Intuitively, since these remaining edges (i.e. those not in the SCCs) form a directed acyclic graph (DAG), calling $\FixAlmostDag(\GB_{\phi_1} \setminus \esep, \{V_1, V_2, \ldots \})$ (\Cref{lem:almostDag}) in Phase 2 produces the following result.

\begin{lemma}\label{thm:ScaleDown:Phase 2 conclusion}
If $G$ has no negative-weight cycle, all weights in $\GB_{\phi_2} \setminus \esep$ are non-negative. 
\end{lemma}
\begin{proof}
Clearly, $\GB_{\phi_1}\setminus \esep$ and $\{V_1, V_2, \ldots \}$ satisfy the input conditions of \Cref{lem:almostDag}, i.e. (1) $(\GB \setminus \esep)_{\phi_1}[V_i]$ contains no negative-weight edges for every $i$ (this is due to \Cref{thm:ScaleDown:Phase 1 conclusion}), and (2) when we contract every $V_i$ into a node, the resulting graph is a DAG (since the $V_i$ are precisely the (maximal) SCCs of $\GB_{\phi_1}\setminus \esep$).\footnote{See, e.g., Lemma 22.13 in \cite{CLRS_book_3rd}.}
Thus, $\FixAlmostDag((\GB \setminus \esep)_{\phi_1}, \{V_1, V_2, \ldots \})$  returns $\psi$ such that $(\GB_{\phi_1} \setminus \esep)_\psi=\GB_{\phi_2} \setminus \esep$ contains no negative-weight edges. 
\end{proof}

\paragraph{Phase 3 runtime.}
Now we are ready to prove \Cref{thm:ScaleDown:main runtime}, i.e. the runtime bound of $\SPaverage((\GB_s)_{\phi_2}, s)$ in Phase 3 when $G$ contains no negative-weight cycle. We start by clarifying the nature of the graph $(\GB_s)_{\phi_2}$. We start with the graph $\GB$; we then get $\GB_s$ by adding a dummy source with outgoing edges of weight $0$ to every $v \in V$; finally we get $(\GB_s)_{\phi_2}$ by applying price function $\phi_2$ to $\GB_s$, where we define $\phi_2(s) = 0$.  The reason we apply $\phi_2$ at the end is to ensure that $\GB_s$ and $(\GB_s)_{\phi_2}$ are equivalent. Note that after we apply $\phi_2$, the edges incident to $s$ can have non-zero weight (positive or negative); but $s$ can still reach every vertex, so we can still call $\SPaverage$ with $s$ as the source.

Recall (\Cref{lem:spaverage,def:eta}) that the runtime of  $\SPaverage((\GB_s)_{\phi_2}, s)$  is
\begin{align*}
    O\left(\left(m+\sum_{v\in V}\eta_{(\GB_s)_{\phi_2}}(v; s)\right)\log m\right) 
\end{align*}
Fix any $v\in V$, observe that
\begin{align*}
\eta_{(\GB_s)_{\phi_2}}(v; s) &= \min\{|P \cap \eneg((\GB_s)_{\phi_2})| : \mbox{$P$ is a shortest $sv$-path in $(\GB_s)_{\phi_2}$}\} & \mbox{(\Cref{def:eta})}\\
&\leq |P_{\GB}(v)\cap \eneg((\GB_s)_{\phi_2})| 
\end{align*}
where the inequality is because, for any price function $\phi_2$, $P_{\GB}(v)$ is a shortest $sv$-path in $(\GB_s)_{\phi_2}$ (because $(\GB_s)_{\phi_2}$ and $\GB_s$ are equivalent and $P_{\GB}(v)$ is a shortest $sv$-path in $\GB_s$ by definition). 
By \Cref{thm:ScaleDown:Phase 2 conclusion}, all negative-weight edges in $\GB_{\phi_2}$ are in $\esep$, i.e. $\eneg(\GB_{\phi_2})\subseteq \esep$; so, 
\[
\eta_{(\GB_s)_{\phi_2}}(v; s)
\leq |P_{\GB}(v)\cap \esep| +1
\]
where the ``+1'' term is because the edge incident to $s$ in $P_{\GB}(v)$ can also be negative.
By \Cref{thm:ScaleDown:expected esep} and the fact that $\eta(\GB)\leq \Delta$ (input requirement \Cref{item:ScaleDown input b} in \Cref{thm:scaledown}),\footnote{The expectation in (\ref{eq:Phase 3 expectation}) is over the random outcomes of the low-diameter decomposition in Phase 0 and the recursion in Phase 1.
Note that both $\eta_{(\GB_s)_{\phi_2}}(v;s)$ and $|P_{\GB}(v)\cap \esep|$ are random variables.
Since we always have $\eta_{(\GB_s)_{\phi_2}}(v;s) \leq |P_{\GB}(v)\cap \esep|+1$, we also have  $E\left[\eta_{(\GB_s)_{\phi_2}}(v;s)\right] \leq E\left[\left|P_{\GB}(v)\cap \esep\right|\right]+1$.} 
\begin{align}
E\left[\eta_{(\GB_s)_{\phi_2}}(v;s)\right] \leq
E\left[\left|P_{\GB}(v)\cap \esep\right| \right]+1=  O(\log^2 m). \label{eq:Phase 3 expectation}    
\end{align}
(The last equality is by \Cref{thm:ScaleDown:expected esep}.) Thus, the expected runtime of $\SPaverage((\GB_s)_{\phi_2}, s)$ is
\begin{align*}
O\left(\left(m+E\left[\sum_{v\in V}\eta_{(\GB_s)_{\phi_2}}(v;s)\right]\right)\log m\right)
& =O\left(m \log^3 m\right). 
\end{align*}

%% file: ScaleDown_pseudo.tex
\begin{algorithm2e}[t]
\caption{Algorithm for $\ScaleDown(G = (V,E,w),\Delta,B)$}\label{alg:ScaleDown}

\tcp*[h]{\textcolor{blue}{Input/Output: See \Cref{thm:scaledown}}}

\If{$\Delta\leq 2$}{
 \label{line:ScaleDown:BaseCase}
Let $\phi_2=0$ and jump to Phase 3 (\Cref{line:ScaleDown:Phase 3}) 
}

Let $d=\Delta/2$. Let $\GB_{\geq 0}:=(V, E, \wB_{\geq 0})$ where
$\wB_{\geq 0}(e) := \max \{0, \wB(e) \}$ for all $e \in E$\label{line:ScaleDown:def}

\BlankLine

\tcp*[h]{\textcolor{blue}{Phase 0: Decompose $V$ to SCCs $V_1, V_2, \ldots$ with weak diameter $dB$ in $G$}}

$\esep \gets \SCCDecomposition(\GB_{\geq 0}, dB)$ (\Cref{lem:SCCDecomposition})
\label{line:ScaleDown:LowDiamDecomposition}

Compute Strongly Connected Components (SCCs) of $\GB \setminus \esep$, denoted by $V_1, V_2, \ldots$ \label{line:ScaleDown:SCC decomposition} \\ \tcp*[h]{Properties: (\Cref{thm:ScaleDown:Weak Diameter}) For each $u,v\in V_i$, $\dist_G(u,v)\leq dB$.} \\ \tcp*[f]{(\Cref{thm:ScaleDown:expected esep}) If $\eta(\GB)\leq \Delta$, then for every $v\in V_i$, $E[P_{\GB}(v)\cap \esep]=O(\log^2 n)$}

\BlankLine

\tcp*[h]{\textcolor{blue}{Phase 1: Make edges inside the SCCs $\GB[V_i]$ non-negative}}

Let $H=\bigcup_i G[V_i]$, i.e. $H$ only contains edges inside the SCCs. \\ \tcp*[f]{(\Cref{thm:phase 1 works}) If $G$ has no negative-weight cycle, then $\eta(H^B)\leq d=\Delta/2$.}

$\phi_1\gets \ScaleDown(H, \Delta/2, B)$ \label{line:ScaleDown:Recursion} \tcp*[f]{(\Cref{thm:ScaleDown:Phase 1 conclusion}) $w_{H^B_{\phi_1}}(e)\geq 0$ for all $e\in H$}

\tcp*[h]{\textcolor{blue}{Phase 2: Make all edges in $\GB \setminus \esep$ non-negative}}

$\psi \gets \FixAlmostDag(\GB_{\phi_1} \setminus \esep, \{V_1, V_2, \ldots \})$ (\Cref{lem:almostDag}) \label{line:ScaleDown:callFixAlmostDAG}

$\phi_{2} \gets \phi_1 + \psi$  \tcp*[f]{(\Cref{thm:ScaleDown:Phase 2 conclusion}) All edges  in $(\GB \setminus \esep)_{\phi_2}$ are non-negative}

\BlankLine

\tcp*[h]{\textcolor{blue}{Phase 3: Make all edges in $\GB$ non-negative}}

$\psi'\gets \SPaverage((\GB_s)_{\phi_2}, s)$ (\Cref{lem:spaverage}) \label{line:ScaleDown:Phase 3} \tcp*[f]{(\Cref{thm:ScaleDown:main runtime}) expected time $O(m \log^3 m)$. (To define $(\GB_s)_{\phi_2}$ here, we define $\phi_2(s)=0$.)}

$\phi_3=\phi_2+\psi'$  \tcp*[f]{(\Cref{thm:ScaleDown:main output correctness}) All edges in $\GB_{\phi_3}$ are non-negative.}

\BlankLine
\BlankLine

\Return $\phi_3$ \tcp*{Since $\wB_{\phi_3}(e)\geq 0$, we have $w_{\phi_3}(e)\geq -B$}

\end{algorithm2e}

%% file: sec_spmain.tex
\section{Algorithm $\SPmain$ (\Cref{thm:spmain})}\label{sec:SPmain}

In this section we present algorithm $\SPmain(\Gin,\sin)$ (\Cref{thm:spmain}), which always runs on the main input graph $\Gin$ and source $\sin$.

\paragraph{Description of Algorithm $\SPmain(\Gin,\sin)$.}
See Algorithm \ref{alg:SPmain} for pseudocode. Recall that if $\Gin$ contains a negative-weight cycle, then the algorithm is not supposed to terminate; for intuition, we recommend the reader focus on the case where $\Gin$ contains no negative-weight cycle.

The algorithm first creates an equivalent graph $\Gbar$ by scaling up edge weights by $2n$ (Line \ref{line:SPMain:what}), and also rounds $B$ (Line \ref{line:SPMain:power}), all to ensure that everything remains integral. It then repeatedly calls $\ScaleDown$ until we have a price function $\phi_t$ such that $w_{\phi_t}(e) \geq -1$ (See for loop in Line \ref{line:SPMain:ForLoop}). The algorithm then defines a graph $\Gstar = (V,E,\wstar)$ with $\wstar(e) = w_{\phi_t}(e) + 1$ (Line \ref{line:SPMain:gstar}). In the analysis, we will argue that because we are dealing with the scaled graph $\Gbar$, the additive $+1$ is insignificant and does not affect the shortest path structure (Claim \ref{claim:spmain:output}), so running Dijkstra on $\Gstar$ will return correct shortest paths in $G$ (Lines \ref{line:SPMain:callDijkstra} and \ref{line:SPMain:Dijkstragstar}). 

\input{spmain_new_pseudocode}

\paragraph{Correctness}
We focus on the case where the algorithm terminates, and hence every line is executed. First we argue that weights in $\Gstar$ (Line \ref{line:SPMain:gstar}) are non-negative. 

\begin{claim}
\label{claim:spmain:price}
If the algorithm terminates, then for all $e \in E$ and $i \in [0,t:=\log_2(B)]$ we have that $\wbar_i$ is integral and that $\wbar_i(e) \geq -B/2^i$ for all $e \in E$. Note that this implies that $\wbar_t(e) \geq -1$ for all $e \in E$, and so the graph $\Gstar$ has non-negative weights.
\end{claim}

\begin{proof}
We prove the claim by induction on $i$. For base case $i = 0$, the claim holds for $\Gbar_{\phi_0} = \Gbar$ because $\win(e) \geq -1$ (see Assumption \ref{assum:main}), so $\wbar(e) \geq -2n \geq -B$ (see Lines \ref{line:SPMain:what} and \ref{line:SPMain:power}). 

Now assume by induction that the claim holds for $\Gbar_{\phi_{i-1}}$. The call to $\ScaleDown(\Gbar_{\phi_{i-1}},\Delta := n,B/2^{i})$ in Line \ref{line:SPMain:callScaleDown} satisfies the necessary input properties (See Theorem \ref{thm:scaledown}): property \ref{prop:scaledown:weight} holds by the induction hypotheses; property \ref{prop:scaledown:Delta} holds because we have $\eta_G(v) \leq n$ for any graph $G$ with no negative-weight cycle; property \ref{prop:scaledown:degree} holds because $\Gin$ has constant out-degree, and the algorithm never changes the graph topology. Thus, by the output guarantee of $\ScaleDown$ we have that $(\wbar_{\phi_{i-1}})_{\psi_i}(e) \geq (B/2^{i-1})/2 = B/2^{i}$. The claim follows because as noted in Definition \ref{def:prelim:price}, $(\wbar_{\phi_{i-1}})_{\psi_i} = \wbar_{\phi_{i-1} + \psi_i} = \wbar_{\phi_i}$.
\end{proof}

\begin{corollary}
\label{cor:spmain:terminate}
If $\Gin$ contains a negative-weight cycle then the algorithm does not terminate.
\end{corollary}

\begin{proof}
Say, for contradiction, that the algorithm terminates; then by \Cref{claim:spmain:price} we have that $\wbar_{\phi_t}(e) \geq -1$. Now, let $C$ be any negative-weight cycle in $\Gin$. Since all weights in $\Gbar$ are multiples of $2n$ (Line \ref{line:SPMain:what}), we know that $\wbar(C) \leq -2n$. But we also know that $\wbar_{\phi_t}(C) \geq -|C| \geq -n$. So $\wbar(C) \neq \wbar_{\phi_t}(C)$, which contradicts Lemma \ref{lem:price-equivalent}.
\end{proof}

Now we show that the algorithm produces a correct output.

\begin{claim}
\label{claim:spmain:output}
Say that $\Gin$ contains no negative-weight cycle. Then, the algorithm terminates, and if $P$ is a shortest sv-path in $\Gstar$ (Line \ref{line:SPMain:gstar}) then it is also a shortest sv-path in $\gin$. Thus, the shortest path tree $T$ of $\Gstar$ computed in Line \ref{line:SPMain:Dijkstragstar} is also a shortest path tree in $\Gin$.
\end{claim}

\begin{proof}
The algorithm terminates because each call to  $\ScaleDown(\Gbar_{\phi_{i-1}},\Delta := n,B/2^{i})$ terminates, because $\Gbar_{\phi_{i-1}}$ is equivalent to $\Gin$ (Lemma \ref{lem:price-equivalent}) and so does not contain a negative-weight cycle.
Now we show that 
\begin{align}
\mbox{$P$ is also a shortest path in $\Gbar_{\phi_t}$}\label{eq:SPMain goal}    
\end{align}
which implies the claim because $\Gbar_{\phi_t}$ and $\Gin$ are equivalent. 
We assume that $s\neq v$ because otherwise the claim is trivial. 
Observe that since all weights in $\Gbar$ are multiples of $2n$ (Line \ref{line:SPMain:what}), all shortest distances are also multiples of $2n$, so for any two $sv$-paths $P_{sv}$ and $P'_{sv}$, $|\wbar(P_{sv}) - \wbar(P'_{sv})|$ is either $0$ or $> n$. It is easy to check that by Lemma \ref{lem:price-equivalent}, we also have that
\begin{align}
\mbox{$|\wbar_{\phi_t}(P_{sv}) - \wbar_{\phi_t}(P'_{sv})|$ is either $0$ or $> n$.} \label{eq:SPmain a}
\end{align}
Moreover, by definition of $\Gstar$ we have 
\begin{align}
    \wbar_{\phi_t}(P_{sv}) < \wstar(P_{sv}) = \wbar_{\phi_t}(P_{sv}) + |P_{sv}| < \wbar_{\phi_t}(P_{sv}) + n.\label{eq:SPmain b}
\end{align}
Now, we prove \eqref{eq:SPMain goal} . Assume for contradiction that there was a shorter path $P'$ in $\Gbar_{\phi_t}$. Then,
\begin{align}
    \wbar_{\phi_t}(P) - \wbar_{\phi_t}(P') \stackrel{\eqref{eq:SPmain a}}{>}n.\label{eq:SPmain c}
\end{align}
So, 
$\wstar(P') \stackrel{\eqref{eq:SPmain b}}{<} \wbar_{\phi_t}(P') + n \stackrel{\eqref{eq:SPmain c}}{<} \wbar_{\phi_t}(P) \stackrel{\eqref{eq:SPmain b}}{<}\wstar(P),$ contradicting $P$ being shortest in $\Gstar$.
\end{proof}

\paragraph{Running Time Analysis:}
By \Cref{cor:spmain:terminate}, if $\Gin$ does not contain a negative-weight cycle then the algorithm does not terminate, as desired. We now focus on the case where $\Gin$ does not contain a negative-weight cycle. The running time of the algorithm is dominated by the $\log(B) = O(\log(n))$ calls to $\ScaleDown(\Gbar_{\phi_{i-1}},\Delta := n,B/2^{i})$. Note that all the input graphs $\Gbar_{\phi_{i-1}}$ are equivalent to $\Gin$, so they do not contain a negative-weight cycle. By Theorem \ref{thm:scaledown}, the expected runtime of each call to $\ScaleDown$ is $O(m\log^3(n)\log(\Delta)) = O(m\log^4(n))$. So, the expected runtime of $\SPmain$ is $O(m\log^5(n))$.


%% file: spmain_new_pseudocode.tex
\begin{algorithm2e}[h]
	\caption{Algorithm for $\SPmain(\gin = (V,E,\win),\sin)$}	\label{alg:SPmain}

$\wbar(e) \gets \win(e) \cdot 2n$ for all $e \in E$,  $\Gbar \gets (V,E,\wbar)$, $B \gets 2n$.\label{line:SPMain:what} \tcp*[f]{scale up edge weights} 

Round $B$ up to nearest power of 2\label{line:SPMain:power}\tcp*[f]{still have $\wbar(e) \geq -B$ for all $e \in E$} 

$\phi_0(v) = 0$ for all $v \in V$ \tcp*[f]{identity price function}

\For(\label{line:SPMain:ForLoop}){$i = 1$ to $t:=\log_2(B)$}{

$\psi_i \gets \ScaleDown(\Gbar_{\phi_{i-1}},\Delta := n,B/2^{i})$ \label{line:SPMain:callScaleDown} 

$\phi_i \gets \phi_{i-1}+\psi_i$ 
\tcp*[f]{(Claim \ref{claim:spmain:price}) $w_{\phi_i}(e)\geq -B/2^{i}$ for all $e\in E$}
\label{line:SPMain:endForLoop} 

}

$\Gstar \gets (V,E,\wstar)$ where $\wstar(e) \gets \wbar_{\phi_t}(e) + 1$ 
for all $e \in E$. \label{line:SPMain:gstar}

\tcp*[h]{Observe: $\Gstar$ in above line has non-negative weights}

Compute a shortest path tree $T$ from $s$ using $\Dijkstra(\Gstar, s)$ (\Cref{thm:prelim:Dijkstra}) \label{line:SPMain:callDijkstra} 

\tcp*[h]{(\Cref{claim:spmain:output}) Will Show: any shortest path in $\Gstar$ is also shortest in $G$} 

\Return shortest path tree $T$ \label{line:SPMain:Dijkstragstar}. 

\end{algorithm2e}

%% file: LDD_new_AB.tex
\section{Algorithm for Low-Diameter Decomposition} \label{sec:low-diam-full}
In this section, we prove Lemma~\ref{lem:SCCDecomposition} which we restate here for convenience. Through this section, we often shorten LowDiamDecomposition$(G,D)$ to $\LDD(G,D)$.
\SCCLemma*

We start with some basic definitions.

\begin{definition}[balls and boundaries]\label{dfn:ball}
Given a directed graph $G = (V,E)$, a vertex $v \in V$, and a distance-parameter $R$, we define $\outgball(v,R) = \{ u \in V \mid \dist(v,u) \leq R \}$. We define $\boundary{\outgball(v,R)} = \{(x,y) \in E \mid x \in \outgball(v,R) \land y \notin \outgball(v,R)\}$.  Similarly, we define $\ingball(v,R) = \{ u \in V \mid \dist(u,v) \leq R \}$ and we define $\boundary{\ingball(v,R)} = \{ (x,y) \in E \mid y \in \ingball(v,R) \land x \notin \ingball(v,R)\}$. We often use $\stargball$ to denote that a ball can be either an out-ball or an in-ball, i.e., $* \in \{in, out\}$.
\end{definition}

\begin{definition}[Geometric Distribution]
\label{def:geometric}
Consider a coin whose probability of heads is $p\in (0,1]$. The geometric distribution $\Geo(p)$ is the probability distribution of the number $X$ of independent coin tosses until obtaining the first heads. We have $\Pr[X = k] = p(1-p)^{k-1}$ for every $k\in \{1, 2, 3, ...\}$. 
\end{definition}

\begin{remark}
\label{remark:n}
In Lemma \ref{lem:SCCDecomposition} and throughout this section, $n$ always refers to the number of vertices in the main graph $\Gin$, i.e. graph in which we are trying to compute shortest paths. This is to ensure that "high probability" is defined in terms of the number of vertices in $\Gin$, rather than in terms of potentially small auxiliary graphs. Note that whenever our shortest path algorithm executes $\LDD(G,D)$ we always have $|V(G)| \leq n$ (we never add new vertices).
\end{remark}
\input{LDD_new_AB_pseudocode}

\paragraph{The algorithm.}
The pseudocode for our low-diameter decomposition algorithm can be found in \Cref{alg:LDD}. Roughly, in Phase~1 each vertex is marked as either {\em in-light}, {\em out-light}, or {\em heavy}. In Phase~2 we repeatedly ``remove'' balls centered at either in-light or out-light vertices: Let $v$ be any in-light vertex (the process is similar for out-light vertices). Consider a ball  $\ingball(v,R_v)$ with radius $R_v$ selected randomly from the geometric distribution $\Geo(p)$, with $p = \min\{1,80\log(n)/D\}$ (\Cref{line:ldd-threshold,line:ldd-ball})\footnote{Throughout this section, $\log(n)$ always denotes $\log_2(n)$}. We add edges pointing into this ball (i.e. edges in $\boundary{\ingball(v, R_v)}$) to $\eboundary$, which will be later added to $\esep.$ We recurse the algorithm on this ball (\Cref{line:ldd-recurse}) which may add more edges to $\esep$ (via $\erecurse$). Finally, we remove this ball from the graph and repeat the process. 
Note that the algorithm may also terminate prematurely with $\esep=E(G)$ in \Cref{line:ldd-terminate,line:ldd-check-diam}. We will show that this happens with very low probability (\Cref{sec:ldd-termination-probability}).

\paragraph{Analysis.} The remainder of this section is devoted to analyze the above algorithm. In \Cref{sec:LDD-observations}, we make a few simple observations. Then, in \Cref{sec:LDD-runtime}, we analyze the runtime of the algorithm. The next two subsections are to bound the probability that an edge is in $\esep$: in \Cref{sec:ldd-termination-probability}  we prove that the probability that the algorithm terminates prematurely on Lines \ref{line:ldd-terminate} or \ref{line:ldd-check-diam} is very small, and complete the bound of $Pr[e\in \esep]$ in \Cref{sec:LDD-e in esep}. 
Finally, in \Cref{sec:LDD-diameter}, we prove the diameter bound, completing the proof of  \Cref{lem:SCCDecomposition}.

\subsection{Observations}\label{sec:LDD-observations}

\begin{observation}
\label{obs:ldd-num-calls}
Consider the execution of an initial call $\LDD(G,D)$ and let $\LDD(G_i,D)$ denote all the lower-level recursive calls that follow from it. The following always holds:
\begin{enumerate}[noitemsep]
    \item For each vertex $v$, there are at most $O(\log(n))$ calls $\LDD(G_i,D)$ such that $G_i$ contains $v.$
    \item The total number of recursive calls $\LDD(G_i,D)$ is $O(n\log(n)).$
\end{enumerate}
\end{observation}

\begin{proof}
The second property follows from the first because $G$ has at most $n$ vertices, and each call $\LDD(G_i)$ contains at least one of those vertices.
For the first property, note that 
{\bf (i)} recursion only occurs in Line \ref{line:ldd-recurse}, 
{\bf (ii)} if a vertex participates in the recursion in Line \ref{line:ldd-recurse}, it will be removed from $G$ on \Cref{line:ldd-remove}, 
{\bf (iii)} each time the algorithm recurses, the number of vertices in the child-graph is at most $7/10$ the number of vertices in the parent-graph because otherwise the algorithm terminates in Line \ref{line:ldd-terminate} without recursing, 
and 
{\bf (iv)} 
the algorithm always terminates when $|V(G)|=1$ because every vertex is heavy in this case, so the while loop is never executed and no recursion occurs.
Thus, a single vertex can participate in only  $\log_{10/7}(n) = O(\log(n))$ recursive calls. 
\end{proof}

\begin{observation}
\label{obs:ldd-sccs}
Consider any call $\LDD(G,D)$ and say that the algorithm makes some recursive call $\LDD(G[\ball^*_G(v,R_v)],D)$ in Line \ref{line:ldd-recurse}.
Then, once the call $\LDD(G,D)$ terminates, if vertex $x \in \ball^*_G(v,R_v)$ and vertex $y \notin \ball^*_G(v,R_v)$, then $x$ and $y$ are not in the same SCC in $G \setminus \esep$.
\end{observation}

\begin{proof}
Say that the recursive call is on an out-ball $\outgball(v,R_v)$; the argument for in-balls is the same. Note that all the edges of $\outgball(v,R_v)$ are added to $\eboundary$ (Line \ref{line:ldd-boundary}) and later to $\esep$ (in Lines \ref{line:ldd-terminate} or \ref{line:ldd-erem-union}).
Thus, in $G \setminus \esep$ there are no edges leaving $\outgball(v,R_v)$, so clearly vertices $x$ and $y$ cannot be in the same SCC.
\end{proof}

\subsection{Running Time Analysis}\label{sec:LDD-runtime}

\begin{lemma}
The total running time of  $\LDD(G,D)$ is $O(|V(G)|\log^3(n) + |E(G)|\log^2(n))$.
\end{lemma}

\begin{proof}
Let us first consider the non-recursive work, i.e. everything except Line \ref{line:ldd-recurse}. The balls in Line \ref{line:ldd_s-balls} are computed via $k = O(\log(n))$ executions of Dijkstra's algorithm in total time $O(|V(G)|\log^2(n) + |E(G)|\log(n))$. The sets in Line \ref{line:ldd_v-balls} are computed from the information in Line \ref{line:ldd_s-balls}. 

Every time the algorithm executes the While loop in Line \ref{line:ldd-main-loop} it first samples $R_v \sim \Geo(p)$ from the geometric distribution; this can be done in $O(\log(n))$ time using e.g. \cite{BringmannF13}.\footnote{By Theorem 5 in \cite{BringmannF13} with $p=\log(n)/D$ and word size $w=\Omega(\log\log D)$, the expected time is $O(\log(\min\{D/\log n, n\})/\log w)=O(\log n).$ Note that there might be other simpler methods since we do not need an exact algorithm.} 
It then computes some ball $\stargball(v,R_v)$ (\Cref{line:ldd-ball}).
All vertices in this ball are then removed from $G$ (Line \ref{line:ldd-remove}), so at this level of recursion each vertex participates in at most one ball. Each ball is computed via Dijkstra's algorithm, which requires $O(\log(n))$ time per vertex explored and $O(1)$ time per edge explored.\footnote{
A standard Dijkstra's implementation requires $O(|V(G)|)$ time to initialize the priority queue and vertices' labels (e.g. \cite{CLRS_book_3rd}). This can be easily modified to avoid the initialization time (e.g. \cite{ThorupZ05}). One way to do this is to initialize the priority queue and labels {\em once} in the beginning of Phase~2. After we use them to compute each ball $\stargball(v,R_v)$ on \Cref{line:ldd-ball}, we reinitialize them at the cost of the number of explored vertices.
}
So, the total time to execute the while loop is $O(|V(G)|\log(n) + |E(G)|)$.

Finally, checking the diameter in Line \ref{line:ldd-check-diam} only requires a single execution of Dijkstra's algorithm.
Thus the overall non-recursive work is $O(|V(G)|\log^2(n) + |E(G)|\log(n))$. By Observation \ref{obs:ldd-num-calls}, each vertex (and thus each edge) participates in $O(\log(n))$ calls, completing the lemma.
\end{proof}

\subsection{Bounding termination probabilities}\label{sec:ldd-termination-probability}

We now show that it is extremely unlikely for the algorithm to terminate in Lines \ref{line:ldd-terminate} or \ref{line:ldd-check-diam}. When we discuss some call $\LDD(G,D)$ in this subsection, we will use $G_0$ to the denote the initial input to $\LDD(G,D)$, in order to differentiate it from the graph $G$ from which vertices can be deleted over the course of the algorithm (Line \ref{line:ldd-remove}). Note that we always have $G \subseteq G_0$.

We start with a useful auxiliary claim. 

\begin{claim}
\label{claim:ldd-correct-markings}
Consider a single call $\LDD(G,D)$ (ignoring all future recursive calls). Then, with probability $\geq 1 -O(1/n^{19})$, the following holds: 
\begin{itemize}
    \item For any vertex $v$ marked in-light we have $|\ingzball(v,D/4)| \leq .7 |V(G_0)|$
    \item For any vertex $v$ marked out-light we have $|\outgzball(v,D/4)| \leq .7 |V(G_0)|$
    \item For any vertex $v$ marked heavy we have $|\ingzball(v,D/4)| > .5 |V(G_0)|$ AND $|\outgzball(v,D/4)| > .5 |V(G_0)|$
\end{itemize}

\end{claim}

\begin{proof}
Consider any single instance of a vertex $v$ being marked during the recursive call. We will show that for this vertex the claim holds with probability $\geq 1-O(1/n^{20})$. A union bound over all $|V(G_0)| \leq n$ vertices completes the lemma. 

We first show that for any vertex $v$ such that $|\ingzball(v,D/4)| > .7 |V(G_0)|$, 
we have $\Pr\left[v \text{ is marked in-light } \right] \leq 1/n^{20}:$
Since $|\ingzball(v,D/4)| > .7 |V(G_0)|$ and since the $k = c\log(n)$ vertices from $S$ are sampled randomly, we know that $\mathbb{E}[|\ingzball(v,D/4) \bigcap S|] > .7k$. But for $v$ to be marked in-light we must have $|\ingzball(v,D/4) \bigcap S| \leq .6k$; a Chernoff bounds shows such a high deviation from expectation to be extremely unlikely.\footnote{We use the following Chernoff's bound. For any independent 0/1 random variables $X_1, \ldots, X_k$, if $X_{avg}:=(1/k) \sum_{i=1}^k X_i$, then $Pr[X_{avg}<\mathbb{E}[X_{avg}]-\epsilon]\leq e^{-2k\epsilon^2}$. Here, we use $\epsilon=0.1$ and define $X_i=1$ iff $s_i\in \stargzball(v,D/4)$ (so, $\mathbb{E}[X_{avg}]=\frac{|\stargzball(v,R_v)|}{|V(G_0)|} > 0.7$).}

By similar proofs, it also holds that (a) for any vertex $v$ such that $|\outgzball(v,D/4)| > .7 |V(G_0)|$, 
we have $\Pr\left[v \text{ is marked out-light } \right] \leq 1/n^{20}$ and (b) for any vertex $v$ such that $|\ingzball(v,D/4)| \leq .5 |V(G_0)|$ or $|\outgzball(v,D/4)| \leq .5 |V(G_0)|$, 
we have $\Pr\left[v \text{ is marked heavy} \right] \leq 1/n^{20}.$
\end{proof}

\begin{claim}
\label{claim:ldd-termination}
Consider a single call $\LDD(G,D)$ (ignoring all future recursive calls). With probability $\geq 1-O(1/n^{19})$,  $\LDD(G,D)$ does not terminate in Lines \ref{line:ldd-terminate} or \ref{line:ldd-check-diam}.
\end{claim}

\begin{proof}
Let us assume that the statement of  \Cref{claim:ldd-correct-markings} holds. Under this assumption the algorithm clearly cannot terminate in the second condition in Line \ref{line:ldd-terminate} (i.e. $\stargball(v,R_v)$ being too large), because $\stargball(v,R_v) \subseteq \stargzball(v,R_v)$. We now show that the algorithm also cannot terminate in Line \ref{line:ldd-check-diam}. Note that by the time we get to Line \ref{line:ldd-check-diam}, the only vertices remaining in $G$ must have been marked heavy, because otherwise the while loop in Line \ref{line:ldd-main-loop} would continue. We now show that if $x$ and $y$ are marked heavy then $\dist_{G_0}(x,y) \leq D/2$, which implies that the algorithm does not terminate in Line \ref{line:ldd-check-diam}. To see that $\dist_{G_0}(x,y) \leq D/2$, recall that we are assuming the statement of \Cref{claim:ldd-correct-markings}, so $\ball_{G_0}^{in}(x,D/4) > |V(G_0)|/2$ and $\ball_{G_0}^{out}(x,D/4) > |V(G_0)|/2,$ so there is some vertex $w$ in the intersection of the two balls, so $\dist_{G_0}(x,y) \leq \dist_{G_0} (x,w) + \dist_{G_0}(w,y) \leq D/4 + D/4 = D/2$. (Note that $w$ might not be marked heavy and also might have been removed from $G$, which is why we only guarantee weak diameter.)

Thus, since \Cref{claim:ldd-correct-markings} holds with probability $\geq 1-O(1/n^{19})$, we know that with this same probability the algorithm does not terminate in Line \ref{line:ldd-check-diam}, and does not terminate in the the second condition of Line \ref{line:ldd-terminate}.

We now need to bound the probability of terminating in the first condition of Line \ref{line:ldd-terminate} (i.e. $R_v>D/4$). Recall that $p = \min \{1,80\log(n)/D\}$ and the definition of the geometric distribution (Definition \ref{def:geometric}). Every time the algorithm samples $R_v \sim \Geo(p)$ in Line \ref{line:ldd-threshold} we have that 
$\Pr[R_v > D/4] = \Pr[\textrm{the first D/4 coins are all tails}] = (1-p)^{D/4} \leq (1-\frac{80\log(n)}{D})^{D/4} < 1/n^{20}$. Since each variable $R_v$ is associated with a specific vertex $v$, a union bound over all $|V(G_0)| \leq n$ vertices completes the proof.
\end{proof}

\subsection{Bounding $Pr[e\in \esep]$}\label{sec:LDD-e in esep}

Observe that $e\in \esep$ if in one of the recursive calls it is either (i) added to $\eboundary$ on \Cref{line:ldd-boundary}, 
or (ii) added to $\esep$ before a recursive call terminates on \Cref{line:ldd-terminate,line:ldd-check-diam}. 
The latter case was shown in the previous subsection (\Cref{claim:ldd-termination}) to happen with probability $O(n^{-19})$ in each recursive call, thus with probability $O(n^{-18}\log(n))$ over all $O(n\log n)$ recursive calls (\Cref{obs:ldd-num-calls}).
Below  (\Cref{cor:ldd-PrEboundary}), we show that the former case happens with probability $O\left(\frac{w(e)\cdot (\log n)^2}{D}\right)$ over all recursive calls. So, the probability that $e$ is added to $\esep$ in any of the recursive calls is $O\left(\frac{w(e)\cdot (\log n)^2}{D}+\frac{\log(n)}{n^{18}}\right)=O\left(\frac{w(e)\cdot (\log n)^2}{D}+{n^{-10}}\right)$ as claimed.

\begin{lemma}\label{lem:ldd-PrEboundary}\label{lem:eboundary}
In a single call of \Cref{alg:LDD} (ignoring all future recursive calls), $$Pr[e\in \eboundary]=O\left(\frac{w(e)\cdot (\log n)}{D}\right).$$
\end{lemma}
\begin{proof}[Proof intuition]
We provide a formal proof of \Cref{lem:eboundary} in Appendix \ref{sec:app-eboundary}. For an intuition, 
note that since $R_v$ is the number of coin tosses until obtaining the first head (\Cref{def:geometric}), 
the algorithm can be viewed as the following {\em ball-growing process:} Start with a ball $\ball^{*}_G(v, 1)$ of radius $1$ at some vertex $v$ and $*\in \{in, out\}$. 
We flip a coin that turns head with probability $p$. Every time we get a tail, we increase the radius of the ball by one. We stop when we get a head. We then put all edges in $\boundary{\stargball(v,R_v)}$ into $\eboundary$ and cut out the ball. We may then repeat the ball growing process from a new vertex. 

To analyze  $Pr[e\in \eboundary]$, consider any edge $(x,y)$. Consider the first time that $x$ is contained in a ball
$B^*_G(v, r)$ (for some $r$) 
during the ball-growing process above.
Observe that if the next $w(x,y)$ coin tosses all return tails, then $(x,y)$ will {\em not} be in $\eboundary$ (because $y$ is either in the ball $B^*_G(v, r+w(x,y))$ or is no longer in $G$). In other words, $Pr[e\in \eboundary]$ is at most the probability that one of the next $w(x,y)$ coin tosses returns head. This is at most $p\cdot w(x,y)=O\left(\frac{w(e)\cdot (\log n)}{D}\right).$ 
\end{proof}

\begin{corollary}\label{cor:ldd-PrEboundary}
For any edge $e$, the probability that \emph{one} of the recursive calls adds $e$ to $\eboundary$ is $O\left(\frac{w(e)\cdot (\log n)^2}{D}\right).$
\end{corollary}
\begin{proof}
By \Cref{obs:ldd-num-calls}, for every edge $e$, there are $O(\log n)$ calls LDD($G_i$, $D$) such that $G_i$ contains $e$. The corollary follows by Union Bound. 
\end{proof}

\subsection{Diameter Analysis}\label{sec:LDD-diameter}

We now show that the set $\esep$ returned by $\LDD(G,D)$ satisfies the first output property of Lemma \ref{lem:SCCDecomposition}. The proof is essentially trivial, since the low-diameter condition is ensured by Line \ref{line:ldd-check-diam}.

\begin{lemma}
For any graph $G$, $\LDD(G,D)$ returns $\esep$ such that each SCC of $G \setminus \esep$ has weak diameter $\leq D$.
\end{lemma}

\begin{proof}
The proof will be induction on $|V(G)|$. The base case where $|V(G)| = 1$ trivially holds. We now assume by induction that the Lemma holds for all $\LDD(G',D)$ with $|V(G')| < |V(G)|$. Note that in the remaining of this proof we let $G$ denote the initial input to $\LDD(G,D)$ and not the graph that changes throughout the execution of the algorithm.

Consider any two vertices $x,y \in V(G)$. We want to show that $\LDD(G,D)$ always returns $\esep$ such that
\begin{align}
\mbox{\em if $x$ and $y$ are in the same SCC in $G \setminus \esep$, then $\dist_G(x,y) \leq D$ and $\dist_G(y,x) \leq D$.}\label{eq:ldd-property-diameter}
\end{align}
Note that if a vertex is in $V(G')$ for recursive call $\LDD(G',D)$ in Line \ref{line:ldd-recurse} then it is immediately removed from the graph $G$ (Line \ref{line:ldd-remove}) and never touched again during the execution of the current call $\LDD(G,D)$. There are thus three possible cases regarding $x$ and $y$.
\begin{description}
    \item[{\bf Case 1}] There is some recursive call $\LDD(G',D)$ in Line \ref{line:ldd-recurse} such that one of $x,y$ is in $G'$ 
    and the other is not.
    \item[{\bf Case 2}] There is some recursive call $\LDD(G',D)$ in Line \ref{line:ldd-recurse} such that both $x$ and $y$ are in $G'$. 
    \item[{\bf Case 3}] Neither $x$ or $y$ participate in any recursive call.
\end{description}

In Case 1, we know by Observation \ref{obs:ldd-sccs} that $x$ and $y$ cannot end up in the same SCC of $G \setminus \esep$, so the property \eqref{eq:ldd-property-diameter} holds.

In Case 2, note that for recursive call  $\LDD(G',D)$ we have that $G'$ is an induced subgraph of $G$, and also that $V(G') \leq .7 V(G)$, because otherwise the algorithm would have terminated in Line \ref{line:ldd-terminate}. Thus, by the induction hypothesis, the Lemma holds for $\LDD(G',D)$. In particular, there are two options: {\bf (1)} If $x$ and $y$ are in the same SCC of $G' \setminus \esep$, then $\dist_G(x,y) \leq \dist_{G'}(x,y) \leq D$ and $\dist_G(y,x) \leq \dist_{G'}(y,x) \leq D$ and {\bf (2)} if $x$ and $y$ are not in the same SCC in $G' \setminus \esep$, then they are also not in the same SCC of $G \setminus \esep$, because by Observation \ref{obs:ldd-sccs}, none of the vertices of $V(G) \setminus V(G')$ can be in the same SCC as either $x$ or $y$ in $G \setminus \esep$.

Finally, Case 3 holds because if $x$ and $y$ don't participate in any recursive calls, then they are not removed from graph $G$. In Line \ref{line:ldd-check-diam} the algorithm checks that all vertices which weren't removed have weak diameter $\leq D$. If the check fails then the algorithm sets $\esep \leftarrow E(G)$, in which case the Lemma trivially holds because all SCCs of $G \setminus \esep$ will be singletons.
\end{proof}

%% file: LDD_new_AB_pseudocode.tex
\begin{algorithm2e}
\caption{Algorithm for $\LDD(G=(V,E),D)$}\label{alg:LDD}

Let $n$ be the global variable in Remark \ref{remark:n} \label{line:ldd-n}

$G_0 \leftarrow G, \esep \leftarrow \emptyset$

\BlankLine
\tcp*[h]{\textcolor{blue}{Phase 1: mark vertices as light or heavy}}

$k \leftarrow c \ln(n)$ for large enough constant $c$

$S \leftarrow \{s_1, \ldots, s_k \}$, where each $s_i$ is a random node in $V$ \tcp*[f]{possible: $s_i = s_j$ for $i \neq j$}

For each $s_i \in S$ compute $\ingball(s_i,D/4)$ and $\outgball(s_i,D/4)$ \label{line:ldd_s-balls}

For each $v \in V$ compute $\ingball(v,D/4) \bigcap S$ and $\outgball(v,D/4) \bigcap S$ using Line \ref{line:ldd_s-balls} \label{line:ldd_v-balls}

\ForEach(\label{line:ldd-marking-loop}){$v \in V$}{

If $|\ingball(v,D/4) \bigcap S| \leq .6 k$, mark $v$ \emph{in-light} \tcp*[f]{whp $|\ingball(v,D/4)| \leq .7|V(G)|$}
    
Else if $|\outgball(v,D/4) \bigcap S| \leq .6k$, mark $v$ \emph{out-light}\tcp*[f]{whp $|\outgball(v,D/4)| \leq .7|V(G)|$}
    
Else mark $v$  \emph{heavy} \tcp*[f]{w.h.p $\ingball(v,D/4) > .5|V(G)|$ and $\outgball(v,D/4) > .5|V(G)|$}
    
}

\BlankLine
\tcp*[h]{\textcolor{blue}{Phase 2: carve out balls until no light vertices remain}}

\While(\label{line:ldd-main-loop}){$G$ contains a node v marked $*$-light for $* \in \{in, out\}$}{

Sample $R_v \sim \Geo(p)$ for $p = \min\{1,80\log_2(n)/D\}$.\label{line:ldd-threshold}

Compute $\stargball(v,R_v)$. \label{line:ldd-ball}

$\eboundary \leftarrow \boundary{\stargball(v,R_v)}$\tcp*[f]{add boundary edges of ball to $\esep$}.\label{line:ldd-boundary}

If $R_v > D/4$ or $|\stargball(v,R_v)| > .7|V(G)|$ then return $\esep \leftarrow E(G)$ and terminate \tcp*[f]{Pr[terminate] $\leq 1/n^{20}$}\label{line:ldd-terminate}

$\erecurse \leftarrow \LDD(G[\stargball(v,R_v)],D)$ \label{line:ldd-recurse}\tcp*[f]{recurse on ball}

$\esep \leftarrow \esep \bigcup \eboundary \bigcup \erecurse$.\label{line:ldd-erem-union}

$G \leftarrow G \setminus \stargball(v,R_v)$ \tcp*[f]{remove ball from $G$} \label{line:ldd-remove}

}

\BlankLine
\tcp*[h]{\textcolor{blue}{Clean Up: check that remaining vertices have small weak diameter in initial input graph $G_0$.}}

Select an arbitrary vertex $v$ in $G$. 

If $\ball_{G_0}^{in}(v,D/2) \nsupseteq V(G)$ or $\ball_{G_0}^{out}(v,D/2) \nsupseteq V(G)$ then return $\esep \leftarrow E(G)$ and terminate \tcp*[f]{Pr[terminate] $\leq 1/n^{20}$. If this does not terminate, then all remaining vertices in $V(G)$ have weak diameter $\leq D$} \label{line:ldd-check-diam}

Return $\esep$

\end{algorithm2e}

%% file: new-appendix-spaverage.tex
\section{Proof of \Cref{lem:spaverage} ($\SPaverage$)}
\label{sec:app:spaverage}

\spaverageLemma*

\begin{algorithm2e}[h] 

\caption{Algorithm for $\SPaverage(G)$}
\label{alg:spaverage}

Set $d(s) \gets 0$ and $d(v) \gets \infty$ for $v \neq s$

Initialize priority queue $Q$ and add $s$ to $Q$.

Initially, every vertex is unmarked 
\label{line:BFDijkstra:init_unmark}

\BlankLine

\tcp*[h]{\textcolor{blue}{Dijkstra Phase}}

\While(\label{line:dijkstra-phase}){Q is non-empty}{

Let $v$ be the vertex in $Q$ with minimum $d(v)$\label{line:djisktra-define-v}

Mark $v$ \label{line:dijkstra-mark}

\ForEach{\label{line:dijkstra-for-loop}edge $(v,x) \in E \setminus \eneg(G)$}{

\If{$d(v) + w(v,x) < d(x)$\label{line:Dijkstra-CheckActive}}{

if $x\notin Q$, add $x$ to $Q$ \label{line:dijkstra-add-Q}

$d(x) \gets d(v) + w(v,x)$ and decrease-key$(Q, x, d(x))$\label{line:dijkstra-decrease-d}

} 
} 
Extract $v$ from $Q$\label{line:dijkstra-extract-v}
} 

\tcp*[h]{\textcolor{blue}{Bellman-Ford Phase}}

\ForEach(\label{line:BF-phase}){marked vertex $v$}{
\ForEach(\label{line:BF-relax}){edge $(v,x) \in E(G)$}{
\If{\label{line:BF-CheckActive}$d(v) + w(v,x) < d(x)$}{

if $x\notin Q$, add $x$ to $Q$ \label{line:BF-add-Q}

$d(x) \gets d(v) + w(v,x)$ and decrease-key$(Q, x, d(x))$\label{line:BF-decrease-d}

} 
} 

Unmark $v$ \label{line:BF-unmark}

} 

If $Q$ is empty: \Return $d(v)$ for each $v\in V$\label{line:iteration-terminate}\tcp*[f]{labels do not change so we have correct distances.}

Go to Line \ref{line:dijkstra-phase} \tcp*[f]{$Q$ is non-empty.}

\end{algorithm2e}

At a high level, we implement $\SPaverage(G,s)$ to compute distance estimates $d(v)$ so that at termination, $d(v) = \dist_{G}(s,v)$ for each $v\in V$;
we will then return $\phi(v) = \dist_{G}(s,v)$. 
The pseudocode is given in \Cref{alg:spaverage}. We use priority queue $Q$ that supports finding an element with a minimum key, extracting an element, and decreasing the key of some element $x$ to $k$ (decrease-key$(Q, x, k)$). It can be implemented as a binary heap to support each queue in $O(\log n)$ time. (This suffices for our needs, although there exist better implementations, such as the Fibonacci heap.) 
Note that, 
instead of extracting $v$ on \Cref{line:dijkstra-extract-v}, one can do so
at the same time when $v$ is marked on \Cref{line:dijkstra-mark}. 
Also, on \Cref{line:BF-relax}, it suffices to only consider edges $(v,x) \in \eneg(G)$. 
We instead consider all edges in $E(G)$ because this implementation simplifies some of the proofs.

\subsection{Correctness}

Throughout the execution of the algorithm, We define an edge $(v,x)$ of $G$ to be \emph{active} if $d(v) + w(v,x) < d(x)$ and \emph{inactive} otherwise. We define an edge $(v,x)$ for which $d(v)=d(x)=\infty$ to be inactive.

\begin{invariant}\label{inv:DijkstraBF:in Q or marked}
            At any point during the execution of Algorithm \ref{alg:spaverage}, if an edge $(v,x)$ is active, then $v$ is in $Q$ or marked.  If $(v,x)$ is active and $(v, x)\in E\setminus \eneg(G)$, then $v$ is in $Q$. 
\end{invariant}
\begin{proof}
The statements hold initially: For all $v\in V\setminus \{s\}$, any edge $(v,x)$ is inactive since $d(v)=\infty$ and $s$ is in $Q.$ The statements are affected when we (i) extract a vertex $v$ from $Q$ (\Cref{line:dijkstra-extract-v}), (ii) unmark a vertex $v$ (\Cref{line:BF-unmark}), or (iii) decrease $d(x)$ for some vertex $x$ (\Cref{line:dijkstra-decrease-d,line:BF-decrease-d}), potentially making some inactive edges $(x,y)$ active. 

For (i), when a vertex $v$ is extracted from $Q$ on \Cref{line:dijkstra-extract-v}, it was already marked on \Cref{line:dijkstra-mark}, so the first statement still holds. Also, all edges $(v, x)\in E\setminus\eneg(G)$ were made inactive using the for-loop (Lines \ref{line:dijkstra-for-loop}-\ref{line:dijkstra-decrease-d}, by setting $d(x)\gets \min\{d(x), d(v) + w(v,x)\}$); thus, the second statement still holds.  
For (ii), before we unmark $v$ on \Cref{line:BF-unmark}, we make all edges $(v, x)$ inactive using the previous for-loop (by setting $d(x)\gets \min\{d(x), d(v) + w(v,x)\}$); thus, the statements still hold right after we unmark $v$. 
For (iii), we add $x$ to $Q$ before decreasing $d(x)$ on \Cref{line:dijkstra-decrease-d,line:BF-decrease-d};
thus the statements still hold after \Cref{line:dijkstra-decrease-d,line:BF-decrease-d} are executed. 
\end{proof}

Observe that the algorithm terminates only when $Q$ is empty and every vertex is unmarked (every vertex is unmarked after the Bellman-Ford Phase). So, we have:

\begin{corollary}
        When the algorithm terminates, all edges are inactive. 
\end{corollary}

For every inactive edge $(v,x)$, $d(v)+w(v,x)-d(x)\geq 0$. 
Thus, $d$ is the desired price function when the algorithm terminates.

\subsection{Running time}

The correctness analysis already implies that if there is a negative cycle then the algorithm never terminates (in particular, it is easy to check that if a graph has a negative-weight cycle, then for any labels $d$ on the vertices there is always contains at least one inactive edge.) So, we assume that there is no negative cycle in this section. 
We analyze each iteration of the Dijkstra and Bellman-Ford phases. The initial iteration is referred to as iteration $0$, the next iteration is iteration $1$, and so on. To simplify notation, we let $\dist(v)$ and $\eta(v)$ refer to $\dist_G(v)$ and $\eta_{G}(v; s)$ for every vertex $v$. The key to the runtime analysis is the following lemma. 

\begin{lemma}\label{thm:DijkstraBF key runtim}
After iteration $i$ of the Dijkstra phase, $d(v)=\dist(s,v)$ for every $v$ where $\eta(v)\leq i$. 
\end{lemma}
\begin{proof}
Iteration 0 of the Dijkstra phase is simply the standard Dijkstra's algorithm on the subgraphs induced by $E\setminus \eneg(G)$; thus, the statement holds for $i=0$ by the standard analysis. We complete the proof by induction.

Fix any $i\geq 1.$
Let $v$ be any vertex such that $\eta(v)=i$. Let $P=(u_0=s, u_1, u_2, \ldots, u_k=v)$ be the shortest $sv$-path with $i$ negative-weight edges. Let $u_j$ is the first vertex in $P$ with $\eta(u)=i$; i.e., we define $u_j$ to be the vertex in $P$ such that $(u_{j-1}, u_j)\in \eneg(G)$ and the subpath $P'$ of $P$ between $u_j$ and $v$ contains no edge in $\eneg(G).$

Consider when we just finished iteration $i-1$ of the Dijkstra phase. By induction hypothesis, we can assume that $d(u_{j-1})=\dist(s,u_{j-1})$ (since $\eta(u_{j-1})\leq i-1$). Observe that after the next execution of the Bellman-Ford phase (i.e. iteration $i-1$),
\begin{align}
d(u_j)&\leq d(u_{j-1})+w(u_{j-1}, u_j)  \label{eq:BFDijkstra:time 2}
\end{align}
This is because, if this does not hold before the Bellman-Ford phase, then $u_{j-1}$ must be marked (by \Cref{inv:DijkstraBF:in Q or marked})\footnote{Note that $Q$ is empty after the Dijkstra phase; so, for every active $(x,y),$ vertex $x$ must be marked after the Dijkstra phase.}, 
and $d(u_j)$ will be set to a value of at most $d(u_{j-1})+w(u_{j-1}, u_j)$ on \Cref{line:BF-decrease-d}. From \eqref{eq:BFDijkstra:time 2}, we can conclude that after iteration $i$ of the Bellman-Ford phase, 
\begin{align}
d(u_j)&\leq d(u_{j-1})+w(u_{j-1}, u_j) = \dist(s, u_j).\label{eq:BFDijkstra:time 1}
\end{align}
Now consider when we finish iteration $i$ of the Dijkstra phase. Since $Q$ is empty and by \Cref{inv:DijkstraBF:in Q or marked}, all edges $(u_j, u_{j+1}), (u_{j+1}, u_{j+2}), \ldots, (u_{k-1}, u_k=v)$ are inactive (recall that they are all in $E\setminus \eneg(G)$ by definition). Thus, 
\[d(v)\leq d(u_j)+w(u_j, u_{j+1})+ w(u_{j+1}, u_{j+2}) + \ldots + w(u_{k-1}, v) \stackrel{\eqref{eq:BFDijkstra:time 1}}{=} \dist(s, v)\] 
It is easy to see that our algorithm always satisfies the invariant $d(v)\geq \dist(s,v)$. Thus $d(v)= \dist(s,v)$ as claimed.
\end{proof}

\Cref{inv:BFDijkstra:marked <= Q} below is needed for proving \Cref{thm:BFDijkstra:marked vertices not add to Q}.

\begin{invariant}\label{inv:BFDijkstra:marked <= Q}
The following always holds during the Dijkstra Phase: For any marked vertex $v$ and any $x\in Q$, $d(v)\leq d(x)$. 
\end{invariant}
\begin{proof}
The invariant trivially holds when the Dijkstra phase begins (\Cref{line:dijkstra-phase}) because there is no marked vertex (the Bellman-Ford phase unmarks all vertices).
The invariant can be affected when 
{\bf (i)} a new vertex $v$ is marked (\Cref{line:dijkstra-mark}), 
{\bf (ii)} a new vertex is added to $Q$ (\Cref{line:dijkstra-add-Q}), or
{\bf (iii)} $d(x)$ decreases for $x\in Q$ (\Cref{line:dijkstra-decrease-d}). 

For {\bf (i)}, marking $v$ on \Cref{line:dijkstra-mark} does not violate the invariant because, by definition, $v$ has the lowest $d(v)$ among vertices in $Q$. 
For {\bf (ii)} (\Cref{line:dijkstra-add-Q}), note that $d(x)>d(v)$ by the if-condition on \Cref{line:Dijkstra-CheckActive} (since $w(v,x)\geq 0$). Moreover, since $v\in Q$, we can conclude by induction that $d(u)\leq d(v)$ for every marked $u$; thus, $d(u)\leq d(v)<d(x)$ and so adding $x$ to $Q$ on \Cref{line:dijkstra-add-Q} does not violate the invariant.
Similarly, for {\bf (iii)} (\Cref{line:dijkstra-decrease-d}), the new value of $d(x)$ is such that $d(x)\geq d(v)$ (since $w(v,x)\geq 0$) and we can conclude by induction that $d(u)\leq d(v)$ for every marked $u$; so, decreasing $d(x)$ on \Cref{line:dijkstra-decrease-d} does not violate the invariant.
\end{proof}

\begin{lemma}\label{thm:BFDijkstra:marked vertices not add to Q}
When a vertex $x$ is added to $Q$ on \Cref{line:dijkstra-add-Q}, it is not marked. 
\end{lemma}
\begin{proof}
Suppose for contradiction that there is a vertex $x$ that is marked when it is added to $Q$ on \Cref{line:dijkstra-add-Q}. This leads to the following contradiction:
{\bf (i)} $d(v)<d(x)$ (due to the if-condition on \Cref{line:Dijkstra-CheckActive}; note that $w(v,x)\geq 0$).
{\bf (ii)} $d(x)\leq d(v)$ (by \Cref{inv:BFDijkstra:marked <= Q}, where $x$ is marked and $v\in Q$). 
\end{proof}

Now we conclude the runtime analysis. 
Observe that if there are $N$ vertices added to $Q$ during the entire execution of the algorithm, then the algorithm's runtime is $O(N \log n)$: each vertex extracted from $Q$ contributes $O(\log n)$ time in the Dijkstra phase and $O(1)$ time in the Bellman-Ford phase; here, we use the fact that every vertex has constant out-degree. It thus suffices to bound $N$, the number of vertices added to $Q$. 

\Cref{thm:BFDijkstra:marked vertices not add to Q} implies that every vertex is added to $Q$ at most once in each iteration of the Dijkstra phase. It is clear from the pseudocode that every vertex is also added to $Q$ at most once in each iteration of the Bellman-Ford phase. 
\Cref{thm:DijkstraBF key runtim} implies that a vertex $v$ can be added to $Q$ only in iterations $0, 1, \ldots, \eta(v)$ of Dijkstra and Bellman-Ford phases; thus, a vertex $v$ can be added to $Q$ at most $2\eta(v)+2$ times, and the number of vertices added to $Q$ overall is $N\leq 2\sum_{v\in V} \eta(v)+2n.$ This implies the claimed runtime of $O(\log(n)\cdot (\sum_{v\in V} \eta(v)+n)).$

%% file: appendix-FixDAG.tex
\section{Proof of \Cref{lem:almostDag} ($\FixAlmostDag$)}\label{sec:app:fixdag}

\begin{algorithm2e}
\caption{Algorithm for $\FixAlmostDag(G = (V,E,w), {\cal P}=\{V_1, V_2, \ldots\})$}\label{alg:SPAlmostDAG}

Relabel the sets $V_1, V_2, \ldots$ so that they are in topological order in $G$. That is, after relabeling, if $(u,v) \in E$, with $u \in V_i$ and $v \in V_j$, then $i \leq j$. 

Define $\mu_j = \min 
{\{ w(u,v) \mid (u,v)\in E^{neg}(G), u \notin V_j, v \in V_j \}} $; here, let $\min\{\emptyset\} = 0$.

\tcp*[f]{$\mu_j$ is min negative edge weight entering $V_j$, or 0 if no such edge exists.}

Define $M_1 \gets \mu_1 = 0$.

\For(\label{line:almostdag-for}\tcp*[f]{make edges into each $V_2, \ldots, V_q$ non-negative}){$j = 2$ to $q$}{
$M_j \gets M_{j-1} + \mu_j$ \tcp*{Note: $M_j = \sum_{k \leq j} \mu_k$} 

Define $\phi(v) \gets M_j$ for every $v\in V_j$ 
}

\Return $\phi$

\end{algorithm2e}

See \Cref{alg:SPAlmostDAG} for pseudocode. Let us first consider the running time. Note that computing each $\mu_j$ requires time proportional to $O(1)$ + [the number of edges in $E$ entering $V_j$]; since the $V_j$ are disjoint, the total time to compute all of the $\mu_j$ is $O(m + n)$. Similarly, it is easy to check that the for loop in Line \ref{line:almostdag-for} only considers each vertex once, so the total runtime of the loop is $O(n)$.

To prove correctness, we need to show that $w_\phi(u,v) \geq 0$ for all $(u,v) \in E$. Say that $u \in V_i$ and $v \in V_j$ and note that because the algorithm labels the sets in topological order, we must have $i < j$. Moreover, by definition of $\mu_j$ we have $\mu_j \leq w(u,v)$. Thus, we have $$w_\phi(u,v) = w(u,v) + \phi(u) - \phi(v) = w(u,v) + M_i - M_j = w(u,v) - \sum_{k=i+1}^j \mu_k \geq  w(u,v) - \mu_j \geq 0.$$

%% file: las-vegas.tex
\section{Proof of \Cref{thm:result} via a Black-Box Reduction}
\label{sec:las-vegas}

\newcommand{\Gplus}{G_{nonneg}}

Recall from the preliminaries that $W_G$ is defined as the absolute value of the most negative edge weight in the graph. Formally, $W_G := \max \{ 2,-\min_{e \in E}\{w(e)\}  \}$. (We take the max with $2$ so that $\log(W_G)$ is well defined.)

In this section, we show how to obtain the Las Vegas algorithm of Theorem~\ref{thm:main-REALLY}. We start by showing how to use $\SPmain$ as a black box to obtain the following Monte Carlo algorithm. 
\begin{theorem}
\label{thm:main}
There exists a randomized algorithm $\SPMonteCarlo(\Gin,\sin)$ that takes $O(m\log^6(n)\log(W_{\Gin}))$ time for an $m$-edge input graph $\Gin$ and source $\sin$ and behaves as follows: 
\begin{itemize}[noitemsep]
    \item if $\Gin$ contains a negative-weight cycle, then the algorithm always returns an error message,
    \item
    if $\Gin$ contains {\em no} negative-weight cycle, then 
    the algorithm returns a shortest path tree from $\sin$ with high probability, and otherwise returns an error message.
\end{itemize}
Note that the algorithm always outputs either an error message, or a (correct) shortest path tree.
\end{theorem}

\begin{proof}
Let $\tspmain$ be the expected runtime of $\SPmain$. The algorithm simply runs $C\log(n)$ versions of $\SPmain$ for some large constant $C$, where each version runs for $2\tspmain$ time steps. If any of the versions returns distances, then $\SPmain$ guarantees that these distances are correct, so we simply return them. If all the versions never terminate, then the algorithm also returns an error. To bound the failure probability, note that if the graph contains a negative-weight cycle then returning an error message is what the algorithm is supposed to do; if it does not contain a negative-weight cycle, then by Markov's inequality each version of $\SPmain$ terminates with probability at least $1/2$, so the probability that none of the versions terminate is at most $1/2^{C\log(n)} = 1/n^C$, as desired.
\end{proof}


The rest of this section presents a Las Vegas algorithm $\SPLasVegas(\Gin,\sin)$ that can return a negative cycle when one exists, and whose running time is $O(m\log^8(n))$ w.h.p when $\Gin$ satisfies the properties of Assumption \ref{assum:main}. For this, we require the following definition

\begin{definition}
Given any graph $G = (V,E,w)$ and any positive integer $B$, let $G^{+B} = (V,E,w^{+B})$ be the same as the graph $G$, except that for all $e \in E$ we have $w^{+B}(e) := w(e) + B$. (Note that unlike in $G^{B}$ of Definition \ref{def:GB}, here we add $B$ to \emph{all} edges, including the positive ones.)
\end{definition}

\paragraph{FindThresh}. The key subroutine of $\SPLasVegas$ is FindThresh, which computes the smallest integer $B\geq 0$ such that no negative cycles exist in $\Gin^{+B}$. This is done using the algorithm $\FindThresh$ of the following lemma. The proof of FindThresh is deferred to Section \ref{sec:FindThresh}.

\begin{lemma}\label{lemma:FindThresh}
Let $H$ be an $m$-edge $n$-vertex graph with integer weights and let $s\in V(H)$. Then there is an algorithm, $\FindThresh(H,s)$ which outputs an integer $B\geq 0$ such that w.h.p.,
\begin{itemize}
\item If $H$ has no negative cycles then $B = 0$, and
\item If $H$ has a negative cycle then $B > 0$, $H^{+(B-1)}$ contains a negative cycle, and $H^{+B}$ does not.
\end{itemize}
The running time of $\FindThresh(H,s)$ is $O(m\log^6(n)\log^2(W_H))$.
\end{lemma}

\begin{definition} If the high probability event of Lemma~\ref{lemma:FindThresh} holds, we refer to $B$ as the \emph{correct} value.
\end{definition}

\subsection{The Las Vegas algorithm}
We now present the algorithm $\SPLasVegas$ mentioned at the beginning of the section, using $\FindThresh$ as a black box.
Recall our assumption that $\Gin$ satisfies $w(e)\geq -1$ for each $e\in \Ein$. Pseudocode for $\SPLasVegas(\Gin,\sin)$ can be found in Algorithm~\ref{alg:SPLasVegas}. For intuition when reading the pseudocode, note that we will show that w.h.p. none of the restart events occur.

\begin{algorithm2e}
\caption{Algorithm $\SPLasVegas(\Gin,\sin)$}\label{alg:SPLasVegas}
Let $G'$ be $\Gin$ with every edge weight multiplied by $n^3$\label{line:GPrime}

$B\gets \FindThresh(G',\sin)$

\If{$B = 0$}
{
  \lIf{$\SPMonteCarlo(\Gin,\sin)$ returns error} {restart $\SPLasVegas(\Gin,\sin)$}\label{line:Restart1}

  Let $T$ be the tree output by $\SPMonteCarlo(\Gin,\sin)$


  \Return {$T$}\label{line:SSSPTreeReturned}
}

\lIf{$\SPMonteCarlo((G')^{+B},\sin)$ returns error} {restart $\SPLasVegas(\Gin,\sin)$}\label{line:Restart3}

Let $\phi(v) = \dist_{(G')^{+B}}(\sin,v)$ for all $v\in V$ be obtained from the tree output by $\SPMonteCarlo((G')^{+B},\sin)$\label{line:PhiFromSPMonteCarlo}

$\Gplus \gets ((G')^{+B})_{\phi}$\label{line:GPlus} \tcp*[f]{(Lemma \ref{lem:price-equivalent}) edge-weights in $\Gplus$ are non-negative}

Obtain the subgraph $G_{{}\leq n}$ of $\Gplus$ consisting of edges of weight at most $n$

\lIf {$G_{{}\leq n}$ is acyclic} {restart $\SPLasVegas(\Gin,\sin)$}\label{line:Restart4}

Let $C$ be an \emph{arbitrary} cycle of $G_{{}\leq n}$ 

\lIf {$C$ is not negative in $\Gin$} {restart $\SPLasVegas(\Gin,\sin)$}\label{line:Restart5}

\Return {$C$}\label{line:CycleReturned}
\end{algorithm2e}

Correctness of $\SPLasVegas(\Gin,\sin)$ is trivial as the algorithm explicitly checks that its output is correct just prior to halting:
\begin{lemma}
If $\SPLasVegas(\Gin,\sin)$ outputs a cycle, that cycle is negative in $\Gin$. If $\SPLasVegas(\Gin,\sin)$ outputs a tree, that tree is a shortest path tree from $\sin$ in $\Gin$.
\end{lemma}

\subsection{Running time}
Since all edge weights in $\Gin$ are $\geq -1$ (Assumption \ref{assum:main}), we have $W_{G'}\leq n^3$. It follows that if there is no restart of the algorithm then by Theorem~\ref{thm:main} and Lemma~\ref{lemma:FindThresh}, its running time is $O(m\log^8(n))$ where the call to $\FindThresh(G',\sin)$ dominates. It remains to show that w.h.p., no restart occurs.

When $B = 0$, Lemma~\ref{lemma:FindThresh} implies that w.h.p., $G'$ and hence $\Gin$ has no negative cycles. By Theorem~\ref{thm:main}, w.h.p.~no restart occurs in Line~\ref{line:Restart1}.

Similarly, when $B > 0$, Lemma~\ref{lemma:FindThresh} implies that w.h.p., $(G')^{+B}$ has no negative cycles, so by Theorem~\ref{thm:main}, w.h.p.~no restart occurs in Line~\ref{line:Restart3}.

To show that w.h.p., no restart occurs in Line~\ref{line:Restart4}, we need Corollary~\ref{corollary:NegCycleInSubgraph} below which shows that $G_{{}\leq n}$ preserves all negative cycles from $(G')^{+(B-1)}$.
\begin{lemma}\label{lemma:NegCycleImpliesSmallWeightCycle}
If Line~\ref{line:PhiFromSPMonteCarlo} is reached and $(G')^{+(B-1)}$ has a negative cycle then that cycle has weight less than $n$ in $\Gplus$.
\end{lemma}

\begin{proof} Edge-weights in $(G')^{+B}$ and $(G')^{+(B-1)}$ differ by at most $1$, so if a cycle is negative in $(G)^{+(B-1)}$ it must have weight at most $n$ in $(G')^{+B}$. The Lemma then follows from the fact that $\Gplus$ and $(G')^{+B}$ are equivalent, so by Lemma \ref{lem:price-equivalent}, the weight of any cycle is the same in  $\Gplus$ and in $(G')^{+B}$.
\end{proof}

\begin{corollary}\label{corollary:NegCycleInSubgraph}
If Line~\ref{line:PhiFromSPMonteCarlo} is reached then every negative cycle in $(G')^{+(B-1)}$ is a cycle in $G_{{}\leq n}$.
\end{corollary}

\begin{proof}
Consider any negative cycle $C$ in $(G')^{+(B-1)}$. By the preceding lemma, we know that $C$ has weight $\leq n$ in $\Gplus$. We also know that all edge-weights of $\Gplus$ are non-negative. This implies the corollary. 
\end{proof}

If a restart occurs in Line~\ref{line:Restart4} then we simultaneously have that $B > 0$ and that $(G')^{+(B-1)}$ has no negative cycles by Corollary~\ref{corollary:NegCycleInSubgraph}, which means that the value of $B$ is not correct; Thus, by Lemma~\ref{lemma:FindThresh}, w.h.p.~no restart occurs in Line~\ref{line:Restart4}.

Finally, Corollary~\ref{corollary:CycleInSubgraphIsNegInG} below implies that w.h.p., no restart occurs in Line~\ref{line:Restart5}.

\begin{lemma}\label{lemma:FindThreshOutput}
If $\FindThresh(G',\sin)$ outputs a correct value $B > 0$ then $B\geq n^2$.
\end{lemma}

\begin{proof}
Since we are assuming that $B > 0$ is correct, it must be the case that $G'$ has some negative-weight cycle $C$, but $(G')^{+B}$ does not. Since all edge-weights in $G'$ are integer multiples of $n^3$, $w_{G'}(C) < 0$ implies that $w_{G'}(C) \leq -n^3$. We know that $w_{(G')^{+B}}(C) \geq 0$, but also by definition of $(G')^{+B}$, we have that $w_{(G')^{+B}}(C) = w_{G'}(C) + B|C| \leq  w_{G'}(C) + Bn$. Combining these inequalities yields $Bn \geq n^3$, so $B \geq n^2$.

\end{proof}

\begin{corollary}\label{corollary:CycleInSubgraphIsNegInG}
If $\FindThresh(G',\sin)$ outputs a correct value $B > 0$, every cycle in $G_{{}\leq n}$ is a negative cycle in $\Gin$.
\end{corollary}

\begin{proof}
Consider a cycle $C$ in $G_{\leq n}$. We have $w_{G_{\leq n}}(C) \leq n^2$ because every edge of $G_{\leq n}$ has weight $\leq n$. We now show that $C$ is a negative-weight cycle in $G'$, which also implies that it is negative in $\Gin$.
Observe that 
\[
w_{(G')^{+B}}(C)=
w_{((G')^{+B})_\phi}(C)=w_{G_{nonneg}}(C)=w_{G_{\leq n}}(C),\]
where 
the first equality is because price functions do not change weights of cycles, 
the second equality is because $G_{nonneg}=((G')^{+B})_\phi$ and the last equality is because  $G_{\leq n}$ is a subgraph of $G_{nonneg}$ with the same edge weights.
It follows that 
$$w_{G'}(C) = w_{(G')^{+B}}(C) - B|C| = w_{G_{\leq n}}(C) - B|C| \leq w_{G_{\leq n}}(C) - 2B \leq n^2-2n^2 < 0,$$ 
as desired. 
(In the inequalities, we use the facts that every cycle has at least two edges and $B \geq n^2$ by the previous lemmas.)

\end{proof}

A union bound over the constant number of restart-lines now shows that w.h.p., no restart occurs. We conclude that w.h.p., $\SPLasVegas(\Gin,\sin)$ runs in time $O(m\log^8(n))$.

\subsection{Proof of Lemma \ref{lemma:FindThresh} -- FindThresh}
\label{sec:FindThresh}

In this section, we describe algorithm $\FindThresh$ from Lemma \ref{lemma:FindThresh}. Recall the definition of $W_H$ from the beginning of Section \ref{sec:las-vegas}. The algorithm (\Cref{alg:FindThresh}) is a simple binary search.

\begin{algorithm2e}[t]
\caption{Algorithm for $\FindThresh(H,s)$}
\label{alg:FindThresh}

$\ell \gets 0$ and $r \gets W_H$

\While(\tcp*[f]{Repeat loop until a value is returned}){TRUE}
{

If $\ell = r$ \Return $r$

$q \gets \lfloor (\ell + r) / 2 \rfloor$ 

Execute $\SPMonteCarlo(H^{+q},s)$.

\If(\tcp*[f]{w.h.p $H^{+q}$ has a negative cycle}){$\SPMonteCarlo$ returned an error}
{
$\ell \gets q+1$ \label{line:findthresh-ell}
}

\Else(\tcp*[f]{$H^{+q}$ does not contain a negative cycle}){
$r \gets q$  \label{line:findthresh-r}
}

}

\end{algorithm2e}

\paragraph{Runtime.}
By a standard binary-search argument, the loop is executed $O(\log(W_H))$ times. Each iteration is dominated by the call to $\SPMonteCarlo$, which takes time $O(m\log^6(n)\log(W_H))$. The total runtime of $\FindThresh$ is thus $O(m\log^6(n)\log^2(W_H))$, as desired.

\paragraph{Correcntess.}
We now show that w.h.p. $\FindThresh(H,s)$ returns a correct value of $B$. 
The only probabilistic component of $\FindThresh$ is $\SPMonteCarlo(G,s)$. Recall the guarantees of $\SPMonteCarlo(G,s)$ from Theorem~\ref{thm:main}: if $G$ contains a negative-weight cycle then the algorithm \emph{always} returns an error message; if $G$ does not contain a negative-weight cycle then with high probability it does not. We say that an execution of $\SPMonteCarlo(G,s)$ is \emph{bad} if $G$ does not contain a negative-weight cycle, but the algorithm returns an error; otherwise we say the execution is \emph{good}.

We know from Theorem \ref{thm:main} that an exeuction of $\SPMonteCarlo$ is good w.h.p. Since $\FindThresh$ only executes $\SPMonteCarlo$ a total of $\log(W_H)$ times, and we will only run $\FindThresh$ on a graph with weights polynomial in $n$, a union bound implies that \emph{every} execution of $\SPMonteCarlo$ is good. We will show that as long as this holds, $\FindThresh$ returns a correct value of $B$. So from now on we assume that every execution of $\SPMonteCarlo$ is good.

We now show that if $\FindThresh(H,s)$ returns $B$, then $G^{+B}$ does not contain a negative-weight cycle, as desired. To see this, note that the algorithm guarantees the invariant that $H^{+r}$ contains no negative-weight cycles: this is true initially because we start with $r = W_H$, so by definition of $W_H$, all edges in $H^{+r}$ are non-negative; it remains true because the algorithm only changes $r$ in Line \ref{line:findthresh-r}, where the else-condition guarantees that $H^{+r}$ contains no negative-weight cycle. Since in the end the algorithm returns $B = r$, we know that $G^{+B}$ does not contain a negative-weight cycle.

Next, we show that if $H$ does contain a negative cycle, then $\FindThresh(H,s)$ returns $B = 0$. This is because $H^{+q}$ will also never contain a negative-weight cycle (because $q$ is always positive), so the algorithm will repeatedly execute the else-statement in Line \ref{line:findthresh-r}, eventually returning $B = \ell = r = 0$.

Finally, we need to show that if $H$ does contain a negative-weight cycle, then $H^{+(B-1)}$ contains a negative-weight cycle. Note that whenever the algorithm changes $\ell$ in Line \ref{line:findthresh-ell}, the preceding if-condition guarantees that $H^{+(\ell-1)}$ contains a negative cycle. The desired claim then follows from the fact that upon termination the algorithm returns $B = \ell = r$.

%% file: appendix-eboundary.tex
\section{Proof of Lemma \ref{lem:eboundary}}
\label{sec:app-eboundary}
We start with some notation that sets up the main argument.

\paragraph{Notation:}
\begin{itemize}
    \item Throughout the proof we fix edge $(u,v)$ and bound $\Pr[(u,v) \in \eboundary]$.
    \item Whenever the algorithm executes the While loop in Line \ref{line:ldd-main-loop}, it picks some light vertex. We will assume it chooses the next vertex to process according to some arbitrary ordering $s_1, s_2, ...$. That is, there is some ordering $s_1, s_2, ...,$ of the vertices that are marked in-light or out-light (Line \ref{line:ldd-marking-loop}), such that every time the algorithm executes the while loop in Line \ref{line:ldd-main-loop}, it does so by picking the first $s_i$ in this ordering that has not yet been removed from $G$. Note that the ordering can be adversarially chosen; our analysis works with any ordering. Note also that once $s_i$ is chosen, the direction of the ball (in-ball or out-ball) is uniquely determined because every vertex only receives one marking (see loop in Line \ref{line:ldd-marking-loop}). 
    \item When the algorithm picks $s_i$, it grows the ball up to radius $R_i:=R_{s_i}$ from $s_i$, where $R_i \sim Geo(p)$ is a random variable. (See Line \ref{line:ldd-threshold}). (Note: $R_1, ...R_n$ are the \emph{only} source of randomness in this proof.) 
    \item We now make a small technical change that does not affect the algorithm but simplifies the analysis. Let $W_{max}$ be the heaviest edge weight in the graph and note that all shortest distances are $\leq (n-1)W_{max}$. Define $R_{max} = nW_{max}$. Note that if $R_i \geq R_{max}$ then $\stargball(s_i,R_i) = \stargball(s_i,R_{max})$. Thus, whenever $R_i > R_{max}$, we instead set $R_i = R_{max}$. This has zero effect on the behaviour of the algorithm or the set $\eboundary$ but will be convenient for the analysis because it ensures that we are now working with a \emph{finite} probabilistic space: there are a finite number of variables $R_i$ and each one is an integer in the bounded set $[1,R_{max}]$.
    \item Recall that $R_i \sim \Geo(p)$ where $p = \min\{1,40\log(n)/D\}$ (Line \ref{line:ldd-threshold}). We will use this variable $p$ in our analysis.
    
    \item We define $G_i:=G_i(R_1...R_{i-1})$ as follows. If $s_i$ is already removed from $G$ when it is considered by the while loop in Line \ref{line:ldd-main-loop} then we define $G_i$ as the empty set. Else, we define $G_i$ to be the graph after \textit{every} $s_j \in \{s_1, s_2, \ldots s_{i-1}\}$ has either been processed by the while loop (i.e. a ball was grown from this $s_j$) or removed from $G$ by the algorithm. Note that $G_i$ is a random variable whose value depends on $R_1,...,R_{i-1}.$
    \item Define $B_i:=\ball^*_{G_i}(s_i, R_{i})$, where the $*$ refers to the direction (in or out) uniquely determined by the choice of $s_i$. 
    $B_i$ is a random variable whose value depends on $R_1, ..., R_{i}$. Note that if $G_i$ is empty then so is $B_i$.
    \item Define $I_i$ to be the event that $u,v \notin B_1, ..., B_{i-1}$ AND $u \in B_i$ if the algorithm grows an out-ball $B_i=\outgball(s_i,R_i)$; otherwise (if the algorithm grows an in-ball $B_i=\ingball(s_i,R_i)$), $I_i$ is defined to be the event that $u,v \notin B_1, ..., B_{i-1}$ AND $v \in B_i.$
    (I stands for ``included''.)
    \item Define $X_i$ to be the event $v \notin B_i$ if the algorithm grows an out-ball $B_i=\outgball(s_i,R_i)$; otherwise (if the algorithm grows an in-ball $B_i=\ingball(s_i,R_i)$), $X_i$ is defined to be the event  $u \notin B_i.$  
    (X stand for ``excluded''.)
\end{itemize}

\begin{observation} \label{obs:eboundary-proof}

\begin{itemize}
    \item $I_i$ and $X_i$ are independent from $R_j$ for $j > i$; in other words, $I_i$ and $X_i$ only depend on $R_1, ..., R_{i-1},R_i$.
    \item The events $I_i$ are disjoint and thus $\sum \Pr[I_i] \leq 1$.
    \item $\Pr[(u,v) \in \eboundary] = \sum_{i\geq 0} Pr_{R_1, ..., R_i}[I_i \land X_i]$.
\end{itemize}
\end{observation}

\begin{proof}
The first property is clear. For the second property, note that if $I_i$ is true then $u \in B_i$, so $u$ will be removed from the graph (Line \ref{line:ldd-remove}) and not be present in any $G_j$, $j > i$, so all $I_j$, $j > i$ are false. For the third property, note that $(u,v)$ is added to $\eboundary$ if and only if for some $i$, both $u$ and $v$ are present in $G_i$ and $(u,v)$ is in the boundary of ball $B_i$ (Line \ref{line:ldd-boundary}), which is precisely captured by $I_i \land X_i$.
\end{proof}

Before proceeding with the proof, we will state a common assumption about probabilistic notation that greatly simplifies the presentation

\begin{assumption}\label{assum:cond-prob}
Given probabilistic events $A,B$, we define $\Pr[A|B] = 0$ if $\Pr[B] = 0$. 
\end{assumption}

This leads to significantly simpler notation because it allows us to write $\Pr[A \land B] = \Pr[B] \cdot \Pr[A|B]$ and separately bound $\Pr[A|B]$ without worrying about undefined conditional probabilities when $\Pr[B]=0$.
This notational simplification is commonly used to state the following law of total probability (which we will use later): If $
\left\{B_1, B_2, \ldots \right\}$ is a finite or countably infinite partition of a sample space (i.e. it is a set of pairwise disjoint events whose union is the entire sample space) and $A$ is another event, then
\begin{align}\label{eq:law of total probability}
    \Pr[A]
    =\sum_i \Pr[A\wedge B_i]  
    = \sum_i \Pr[A\mid B_i]\Pr[B_i].  
\end{align}
(Without \Cref{assum:cond-prob} the sum on the right-hand side and its applications below must be only over $i$ such that $\Pr[B_i]>0$.)
The assumption is justified because we are dealing with a finite probability space, so all zero-probability events combined still have zero probability mass.\footnote{By contrast, in an infinite space one would have to be more careful about such an assumption: it is possible for \textit{every} individual instantiation of the random variables to have probability $0$, but all infinity of them combined to have probability mass $1$.} 

This notational assumption allows us to state the following lemma.

\begin{lemma}\label{lem:eboundary-proof}
For any $i$, $\Pr_{R_1, \ldots, R_{i}}[X_i | I_i] \leq pw(u,v)$. (Recall that $p = \min\{1,40\log(n)/D\}$.)
\end{lemma}

We first show that Lemma  \ref{lem:eboundary-proof} completes the proof:

\begin{proof}[Proof of Lemma \ref{lem:eboundary} given Lemma \ref{lem:eboundary-proof}]
By Observation \ref{obs:eboundary-proof} and \eqref{eq:law of total probability} we have\footnote{Here we use \Cref{assum:cond-prob}; otherwise, some of the sums must be over all $i$ such that $\Pr[I_i]>0$.} 
\begin{align*}
   \Pr[(u,v) \in \eboundary]  & = \sum_{i\geq 0} Pr_{R_1, ..., R_i}[I_i \land X_i] 
  && = \sum_{i \geq 0} \Pr[I_i] \cdot \Pr[X_i | I_i] \\
  &\leq p w(u,v) \cdot \sum_{i\geq 0} \Pr[I_i] &&\leq p w(u,v) = O(\log(n) \cdot w(u,v) / D).\qedhere
\end{align*}
\end{proof}

To prove Lemma \ref{lem:eboundary-proof} we use the following claim, which is the crux of the analysis. The claim allows us to fix the first $i-1$ random variables $R_1,\ldots, R_{i-1}$ and only treat $R_i$ as random.

\begin{claim}\label{claim:eboundary-proof} 
Fix \emph{any} instantiations of the first $i-1$ random variables $R_1=r_1, ..., R_{i-1}=r_{i-1}$ such that the resulting graph $G_i = G(r_1, ..., r_{i-1})$ contains both $u$ and $v$. Then, $Pr_{R_i}[X_i|I_i] \leq p w(u,v)$.
\end{claim}

\begin{proof}[Proof of Lemma \ref{lem:eboundary-proof} assuming Claim \ref{claim:eboundary-proof}]

Applying the law of total probability \eqref{eq:law of total probability} over all possible instantiations $r_1, ..., r_{i-1}$ of $R_1, ..., R_{i-1}$ we have 
\begin{equation}\label{eq:full-sum}
\Pr_{R_1, ..., R_{i}}[X_i | I_i] = \sum_{r_1, ..., r_{i-1}} \Pr_{R_i}[X_i | I_i \land R_1 = r_1 \land \ldots \land R_{i-1} = r_{i-1}] \cdot \Pr[R_1 = r_1 \land \ldots \land R_{i-1} = r_{i-1}] 
\end{equation}

To bound the above, we consider two cases. If $R_1 = r_1, \ldots, R_{i-1} = r_{i-1}$ are such that $G_i = G(r_1, ..., r_{i-1})$ does not contain $u$ or does not contain $v$ then by definition $I_i$ is false, so $Pr_{R_i}[X_i | I_i \land R_1 = r_1 \land \ldots \land R_{i-1} = r_{i-1}] = 0$. (Here we are using Assumption \ref{assum:cond-prob}.) The second case is that $G_i = G(r_1, ..., r_{i-1})$ contains both $u$ and $v$, in which case we can apply Claim \ref{claim:eboundary-proof}. Thus, in either case, we have
\begin{equation} \label{eq:cond-prob}
\Pr_{R_i}[X_i | I_i \land R_1 = r_1 \land \ldots \land R_{i-1} = r_{i-1}] \leq p w(u,v)
\end{equation}

Combining (\ref{eq:full-sum}) and (\ref{eq:cond-prob}) we have
$$\Pr_{R_1, ..., R_{i}}[X_i | I_i] \leq p w(u,v) \sum_{r_1, ..., r_{i-1}} \Pr[R_1 = r_1 \land \ldots \land R_{i-1} = r_{i-1}] = p w(u,v).\qedhere$$
\end{proof}

All that remains is to prove Claim \ref{claim:eboundary-proof}.

\begin{proof}[Proof of Claim \ref{claim:eboundary-proof}]

Recall that in this claim we have a fixed graph $G_i$ that contains both $u$ and $v$ and that the only randomness now comes from $R_i \sim \Geo(p)$. Note also that the LDD algorithm (Algorithm \ref{alg:LDD}) is removing vertices and edges from the graph $G$ over time, and that by when the algorithm reaches $s_i$, the remaining graph $G$ that it is working on is by definition $G = G_i$.

For the rest of this proof, we will assume that the algorithm grows an out-ball $\outgball(s_i,R_i)$; the case for an in-ball is exactly analogous. We thus have: 
\begin{align*}
\Pr[X_i | I_i] &= \Pr[v \notin \outgball(s_i,R_i)) \mid u \in \outgball(s_i,R_i))  \\ 
&= \Pr[R_i < \dist_G(s,v) \mid R_i \geq \dist_G(s,u)] \\ 
&\leq \Pr[R_i < \dist_G(s,u) + w(u,v) | R_i \geq \dist_G(s,u)]
\end{align*}
We can bound the last line by observing that the geometric distribution observes the memorylessness property: for any $j > i$ we have $Pr[R_i \leq i+j | R_i\geq i]\leq Pr[R_i \leq j]$. Combining this with the equation above we have
\begin{equation}
\label{eqn:use-memoryless}
\Pr[X_i \mid I_i] \leq \Pr[R_i < \dist_G(s,u) + w(u,v) \mid R_i \geq \dist_G(s,u)] \leq \Pr[R_i \leq w(u,v)] \leq pw(u,v),
\end{equation}
where the last inequality follows from a simple analysis of the geometric distribution: each coin is head with probability $p$, so $\Pr[R_i \leq w(u,v)]$ is the probability that one of the first $w(u,v)$ coins is a head, which by the union bound is at most $pw(u,v)$.

\begin{remark}
The second inequality of (\ref{eqn:use-memoryless}) used the memorylessness property of the geometric distribution, but we have to be a bit careful because recall that in our analysis, in order to maintain a finite probability space, we argued the algorithm sampling $R_i \sim \Geo(p)$ is equivalent to the algorithm sampling  $R_i \sim \Geo(p)$ but then rounding down to $R_i = R_{max} = nW_{max}$ if $R_i > R_{max}$. This means that we can only apply the memorlyessness property in Equation (\ref{eqn:use-memoryless}) if $\dist_G(s,u) + w(u,v) \leq R_{max}$. Fortunately, this is indeed the case because $W_{max}$ is the maximum edge weight, so $w(u,v) \leq W_{max}$ and $\dist_G(s,u) \leq (n-1)W_{max}$.\qedhere
\end{remark}
\end{proof}

We have thus completed the proof of Lemma \ref{lem:eboundary}.